\documentclass[submission,copyright,creativecommons]{eptcs}
\providecommand{\event}{LSFA 2023} 
\usepackage{underscore}           
\usepackage{amsmath}
\usepackage{amsthm}
\usepackage{amsfonts}
\usepackage{amssymb}
\usepackage{enumerate}
\usepackage{hyperref}
\usepackage{quiver}
\usepackage[all,cmtip]{xy}
\usepackage{proof}

\newtheorem{theorem}{Theorem}[section]

\newtheorem{lemma}[theorem]{Lemma}

\newtheorem{proposition}[theorem]{Proposition}
\theoremstyle{definition}

\makeatletter
\newsavebox{\@brx}
\newcommand{\llangle}[1][]{\savebox{\@brx}{\(\m@th{#1\langle}\)}%
  \mathopen{\copy\@brx\kern-0.5\wd\@brx\usebox{\@brx}}}
\newcommand{\rrangle}[1][]{\savebox{\@brx}{\(\m@th{#1\rangle}\)}%
  \mathclose{\copy\@brx\kern-0.5\wd\@brx\usebox{\@brx}}}
\makeatother

\newcommand{\tl}{\otimes \mathsf{L}}
\newcommand{\tr}{\otimes \mathsf{R}}

\newcommand{\lright}{{\multimap}\mathsf{R}}
\newcommand{\lleft}{{\multimap}\mathsf{L}}
\newcommand{\pass}{\mathsf{pass}}
\newcommand{\unitl}{\mathsf{IL}}
\newcommand{\unitr}{\mathsf{IR}}

\newcommand{\andlone}{\land \mathsf{L}_{1}}
\newcommand{\andltwo}{\land \mathsf{L}_{2}}
\newcommand{\andli}{\land \mathsf{L}_{i}}
\newcommand{\andr}{\land \mathsf{R}}
\newcommand{\orl}{\lor \mathsf{L}}
\newcommand{\orrone}{\lor \mathsf{R}_{1}}
\newcommand{\orrtwo}{\lor \mathsf{R}_{2}}
\newcommand{\orri}{\lor \mathsf{R}_{i}}
\newcommand{\ax}{\mathsf{ax}}

\newcommand{\ot}{\otimes}
\newcommand{\lolli}{\multimap}

\newcommand{\I}{\mathsf{I}}

\newcommand{\C}{\mathsf{C}}
\newcommand{\RI}{\mathsf{RI}}
\newcommand{\LI}{\mathsf{LI}}

\newcommand{\F}{\mathsf{F}}
\newcommand{\sw}{\mathsf{sw}}
\newcommand{\tP}{\mathbb{P}}
\newcommand{\tCone}{\mathbb{C}_1}
\newcommand{\tCtwo}{\mathbb{C}_2}
\newcommand{\tE}{\mathbb{R}}
\newcommand{\tT}{\mathbb{T}}

\newcommand{\ex}{\mathsf{ex}}
\newcommand{\topr}{\top \mathsf{R}}
\newcommand{\botl}{\bot \mathsf{L}}
\newcommand{\conj}[1]{\mathsf{conj} (#1)}
\newcommand{\impconj}[1]{\mathsf{impconj} (#1)}

\newcommand{\spl}{\raisebox{-1.5pt}[0.5\height]{\hspace{1pt}\vdots\hspace{1pt}}}

\newcommand{\proofbox}[1]{\begin{tabular}{l} #1 \end{tabular}}

\newcommand{\FSkMCC}{\mathsf{FDSkM}}

\title{Semi-Substructural Logics with Additives}
\author{
Niccol{\`o} Veltri \qquad\qquad Cheng-Syuan Wan
\institute{Tallinn University of Technology, Estonia}
\email{niccolo@cs.ioc.ee \qquad\qquad cswan@cs.ioc.ee}
}
\def\titlerunning{Semi-Substructural Logics with Additives}
\def\authorrunning{N. Veltri \& C.-S. Wan}
\begin{document}
\maketitle
\begin{abstract}
This work concerns the proof theory of (left) skew monoidal categories and their variants (e.g. closed monoidal, symmetric monoidal), continuing the line of work initiated in recent years by Uustalu et al.
Skew monoidal categories are a weak version of Mac Lane's monoidal categories, where the structural laws are not required to be invertible, they are merely natural transformations with a specific orientation. 
Sequent calculi which can be modelled in such categories can be identified as deductive systems for restricted substructural fragments of intuitionistic linear logic. These calculi enjoy cut elimination and admit a focusing strategy, sharing resemblance with Andreoli's normalization technique for linear logic. The focusing procedure is useful for solving the coherence problem of the considered categories with skew structure.

Here we investigate possible extensions of the sequent calculi of Uustalu et al.~with additive connectives. 
As a first step, we extend the sequent calculus with additive conjunction and disjunction, corresponding to studying the proof theory of skew monoidal categories with binary products and coproducts satisfying a left-distributivity condition. 
We introduce a new focused sequent calculus of derivations in normal form, which employs tag annotations to reduce non-deterministic choices in bottom-up proof search.
The focused sequent calculus and the proof of its correctness have been formalized in the Agda proof assistant.
We also discuss extensions of the logic with additive units, a form of skew exchange and linear implication.
\end{abstract}

\section{Introduction}

Substructural logics are logical systems in which the usage of one or more structural rules is disallowed or restricted. A well-known example is the syntactic calculus of Lambek \cite{lambek:mathematics:58}, in which all the structural rules of exchange, weakening and contraction are disallowed. Variants of the Lambek calculus allow exchange or a cyclic form of exchange, while others disallow even associativity \cite{moot:logic:12}. In Girard's linear logic, which has been studied both in the presence and absence of an exchange rule \cite{girard:linear:87,abrusci:noncommutative:1990}, selective versions of weakening and contraction can be recovered via the use of modalities. Applications of substructural logics are abundant in a variety of different fields, from computational investigations of natural languages to the design of resource-sensitive programming languages.

In recent years, in collaboration with Tarmo Uustalu and Noam Zeilberger, we initiated a program intended to study a family of \emph{semi-substructural} logics, inspired by developments in category theory by Szlach{\'a}nyi, Street, Bourke, Lack and many others \cite{szlachanyi:skew-monoidal:2012,lack:skew:2012,street:skew-closed:2013,lack:triangulations:2014,buckley:catalan:2015,bourke:skew:2017,bourke:skew:2018,bourke:lack:braided:2020}. Korn{\'e}l Szlach\'anyi introduced \emph{skew monoidal categories} as a weakening of MacLane's monoidal categories in which the structural morphisms of associativity and unitality (often also called associator and unitors) are not required to be invertible, they are just natural transformation in a particular direction. As such, they can be regarded as \emph{semi-associative} and \emph{semi-unital} variants of monoidal categories. Bourke and Lack also introduced notions of braiding and symmetry for skew monoidal categories which involve three objects instead of two \cite{bourke:lack:braided:2020}. Skew monoidal categories arise naturally in semantics of programming languages \cite{altenkirch:monads:2014} and semi-associativity has strong connections with combinatorial structures such as the Tamari lattice and Stasheff associahedra \cite{zeilberger:semiassociative:19,moortgat:tamari:20}.

Semi-substructural logics correspond to the internal languages of skew monoidal categories and their extensions, therefore sitting in between (certain fragments of) non-associative and associative intuitionistic linear logic. Semi-associativity and semi-unitality can be hard-coded in the sequent calculus following a two-step recipe. First, consider sequents of the form $S \mid \Gamma \vdash A$, where the antecedent is split into an optional formula $S$, called stoup, and an ordered list of formulae $\Gamma$. The succedent consists of a single formula $A$. Then restrict the application of introduction rules to allow only one of the directions of associativity and unitality, the one in the definition of skew monoidal category. For example, left-introduction rules are allowed to act only on the formula in stoup position, not on formulae in $\Gamma$.

In our previous investigations we have explored deductive systems for $(i)$ skew semigroup \cite{zeilberger:semiassociative:19}, $(ii)$ skew monoidal \cite{uustalu:sequent:2021}, $(iii)$ skew (prounital) closed \cite{uustalu:deductive:nodate} and $(iv)$ skew monoidal closed categories \cite{UVW:protsn,veltri:multifocus:23}, corresponding to skew variants of the fragments of non-commutative intuitionistic linear logic consisting of connectives $(i)$ $\otimes$, $(ii)$ $(\I,\otimes)$, $(iii)$ $\lolli$ and $(iv)$ $(\I,\otimes,\lolli)$. We have also studied partial normality conditions, when one or more among associator and unitors is allowed to have an inverse \cite{uustalu:proof:nodate}, and extensions with exchange {\`a} la Bourke and Lack \cite{veltri:coherence:2021}.

When studying meta-theoretic properties of these semi-structural deductive systems, we have been mostly interested in categorical and proof-theoretic semantics. In the latter, we have particularly investigated normalization strategies inspired by Andreoli's focused sequent calculus for classical linear logic \cite{andreoli:logic:1992} and employed the resulting normal forms to solve the \emph{coherence problem} for the corresponding categories with skew structure. For these categories, the word problem is more nuanced than in the normal non-skew case studied by MacLane \cite{maclane1963natural}. Our study additionally revealed that the focused sequent calculi of semi-substructural logics can serve as cornerstones for a compositional and modular understanding of normalization techniques for other richer substructural logics.

In this work we begin the investigation of semi-substructural logics with \emph{additive connectives}. We start in Section \ref{sec:sequent-calculus} by considering a fragment of non-commutative linear logic consisting of skew multiplicative unit $\I$ and conjunction $\ot$, and additive conjunction $\land$ and disjunction $\lor$. We describe a cut-free sequent calculus and a congruence relation $\circeq$ identifying derivations up-to $\eta$-equivalence and permutative conversions. In Section \ref{sec:categorical}, we discuss categorical semantics in terms of skew monoidal categories with binary products and coproducts satisfying a left-distributivity condition.

In Section \ref{sec:focusing}, we introduce a sequent calculus of proofs in normal form, i.e. canonical representative of the equivalence relation on derivations $\circeq$. The design of the latter calculus is again inspired by the ideas of Andreoli and it describes a sound and complete root-first proof search strategy for the original sequent calculus. Completeness is achieved by marking sequents with lists of \emph{tags}, a mechanism introduced by Uustalu et al. \cite{UVW:protsn} and inspired by Scherer and R{\'e}my's saturation technique \cite{scherer:simple:2015}, which helps to completely eliminate all undesired non-determinism in proof search and faithfully capture normal forms wrt. the congruence relation on derivations in the original sequent calculus.

To showcase the modularity of our normalization strategy, in Section \ref{sec:extensions} we discuss extensions of the logic with other connectives, such as additive units, the structural rule of exchange in the style of Bourke and Lack and linear implication.
This provides evidence that our normalization technique is potentially scalable to other richer substructural logics arising as extensions of ours, e.g. full Lambek calculus or intuitionistic linear logic.
Moreover, we conjecture that a similar use of tags can be ported to fragments of classical linear logic, such as MALL.
The resulting notion of normal form should correspond to maximally multi-focused proofs \cite{chaudhuri:canonical:2008} and therefore proof nets.

The sequent calculi of Sections \ref{sec:sequent-calculus} and \ref{sec:focusing}, as well as the effective normalization procedure, have been fully formalized in the Agda proof assistant. The code is freely available at
\begin{center}
  \url{https://github.com/cswphilo/SkewMonAdd}.
\end{center}

\section{Sequent Calculus}\label{sec:sequent-calculus}
We start by describing a sequent calculus for a skew variant of non-commutative multiplicative intuitionistic linear logic with additive conjunction and disjunction.

Formulae are inductively generated by the grammar $A,B ::= X \ | \ \I \ | \ A \ot B \ | \ A \land B \ | \ A \lor B$, where $X$ comes from a set $\mathsf{At}$ of atomic formulae. 
We use $\I , \ot , \land$ and $\lor$ to denote multiplicative verum, multiplicative conjunction, additive conjunction and additive disjunction, respectively. The additives are traditionally named $\&$ and $\oplus$ in linear logic literature.

A sequent is a triple of the form $S \mid \Gamma \vdash A$.
The antecedent is split in two parts: an optional formula $S$, called \emph{stoup} \cite{girard:constructive:91}, and an ordered list of formulae $\Gamma$, called \emph{context}.
The succedent $A$ is a single formula.
The peculiar design of sequents, involving the presence of the stoup in the antecedent, comes from our previous work on deductive systems with skew structure in collaboration with Uustalu and Zeilberger \cite{uustalu:sequent:2021,uustalu:proof:nodate,uustalu:deductive:nodate,veltri:coherence:2021,UVW:protsn,veltri:multifocus:23}.
The metavariable $S$ always denotes a stoup, i.e., $S$ can be a single formula or empty, in which case we write $S = {-}$. Metavariables $X,Y,Z$ are always names of atomic formulae.

Derivations of a sequent $S \mid \Gamma \vdash A$ are inductively generated by the following rules:

\begin{equation}\label{eq:seqcalc}
  \def\arraystretch{2.5}
  \begin{array}{c}
    \infer[\ax]{A \mid \quad \vdash A}{}
    \qquad
    \infer[\pass]{{-} \mid A , \Gamma \vdash C}{A \mid \Gamma \vdash C}
    \qquad
    \infer[\unitl]{\I \mid \Gamma \vdash C}{{-} \mid \Gamma \vdash C}
    \qquad
    \infer[\unitr]{{-} \mid \quad \vdash \I}{}
    \\
    \qquad
    \infer[\tl]{A \ot B \mid \Gamma \vdash C}{A \mid B , \Gamma \vdash C}
    \qquad
    \infer[\tr]{S \mid \Gamma , \Delta \vdash A \ot B}{
      S \mid \Gamma \vdash A
      &
      {-} \mid \Delta \vdash B
    }
    \\
    \infer[\andlone]{A \land B \mid \Gamma \vdash C}{A \mid \Gamma \vdash C}
    \qquad
    \infer[\andltwo]{A \land B \mid \Gamma \vdash C}{B \mid \Gamma \vdash C}
    \qquad
    \infer[\andr]{S \mid \Gamma \vdash A \land B}{
      S \mid \Gamma \vdash A
      &
      S \mid \Gamma \vdash B
    }
    \\
    \infer[\orl]{A \lor B \mid \Gamma \vdash C}{
      A \mid \Gamma \vdash C
      &
      B \mid \Gamma \vdash C
    }
    \qquad
    \infer[\orrone]{S \mid \Gamma \vdash A \lor B}{S \mid \Gamma \vdash A}
    \qquad
    \infer[\orrtwo]{S \mid \Gamma \vdash A \lor B}{S \mid \Gamma \vdash B}
  \end{array}
\end{equation}
The inference rules are similar to the ones in non-commutative intuitionistic linear logic \cite{abrusci:noncommutative:1990}, but with some essential differences. 
\begin{enumerate}
\item The left logical rules $\unitl$, $\tl$, $\andli$ and $\orl$, when read bottom-up, can only be applied on the formula in the stoup position. 
That is, it is generally not possible to remove a unit $\I$, or decompose a tensor $A \ot B$ or a disjunction $A \lor B$, when these formulae are located in the context.
\item The right tensor rule $\tr$, when read bottom-up, splits the antecedent of the conclusion but the formula in the stoup, whenever this is present, always moves to the left premise.
The stoup formula of the conclusion is prohibited to move to the second premise even if $\Gamma$ is empty. 
\item The presence of the stoup implies a distinction between antecedents of the form $A \mid \Gamma$ and ${-} \mid A, \Gamma$. The structural rule $\pass$ (for `passivation'), when read bottom-up, allows the moving of the leftmost formula in the context to the stoup position whenever the stoup is initially empty.
\end{enumerate}
These restrictions allow the derivability of sequents $(A \ot B) \ot C \mid ~\vdash A \ot (B\ot C)$ (semi-associativity), $\I \ot A \mid ~ \vdash A$ and $A \mid ~ \vdash A \ot \I$ (semi-unitality), while forbidding the derivability of their inverses, where the formulae in the stoup and in the succedent have been swapped. This is in line with the intended categorical semantics, see Section \ref{sec:categorical}.

Notice that, similarly to the case of non-commutative intuitionistic linear logic \cite{abrusci:noncommutative:1990}, all structural rules of exchange, contraction and weakening are absent. We give names to derivations and we write $f : S \mid \Gamma \vdash A$ when $f$ is a particular derivation of the sequent $S \mid \Gamma \vdash A$.

\begin{theorem}
  The sequent calculus enjoys cut admissibility: the following two cut rules are admissible
    \begin{displaymath}
      \infer[\mathsf{scut}]{S \mid \Gamma , \Delta \vdash C}{
        S \mid \Gamma \vdash A
        &
        A \mid \Delta \vdash C
      }
      \qquad
      \infer[\mathsf{ccut}]{S \mid \Delta_0 , \Gamma , \Delta_1 \vdash C}{
        {-} \mid \Gamma \vdash A
        &
        S \mid \Delta_0 , A , \Delta_1 \vdash C
      }
    \end{displaymath}
\end{theorem}
  
While the left $\land$-rules only act on the formula in stoup position (as all the other left logical rules), other $\land$-rules $\andli^{\mathsf{C}}$ acting on formulae in context are admissible.
  \begin{displaymath}
    \begin{array}{c}
      \proofbox{
        \infer[\andlone^{\mathsf{C}}]{S \mid \Gamma , A \land B , \Delta \vdash C}{
          S \mid \Gamma , A ,\Delta \vdash C
        }
      }
      \quad
      \proofbox{
        \infer[\andltwo^{\mathsf{C}}]{S \mid \Gamma , A \land B , \Delta \vdash C}{
          S \mid \Gamma , B ,\Delta \vdash C
        }
      }
    \end{array}
  \end{displaymath}
However, this is not the case for the other left logical rules. For example, there is no way of constructing a general left $\lor$-rule $\orl^{\mathsf{C}}$ acting on a disjunction in context. This rule should be forbidden since it would make some inadmissible sequents provable in the sequent calculus. 
For example, the sequent $X \land Y \mid Y \lor X \vdash (X \ot Y) \lor (Y \ot X)$ is not admissible (this can be proved using the normalization procedure of Section \ref{sec:focusing}) but a proof could be found using $\orl^{\mathsf{C}}$:
\begin{displaymath}
  \footnotesize
  \begin{array}{c}
  \infer[\orl^{\mathsf{C}}]{X \land Y \mid Y \lor X \vdash (X \ot Y) \lor (Y \ot X)}{
    \infer[\orrone]{X \land Y \mid Y \vdash (X \ot Y) \lor (Y \ot X)}{
      \infer[\andlone]{X \land Y \mid Y \vdash X \ot Y}{
        \infer[\tr]{X \mid Y \vdash X \ot Y}{
          \infer[\ax]{X \mid \quad \vdash X}{}
          &
          \infer[\pass]{- \mid Y \vdash Y}{
            \infer[\ax]{Y \mid \quad \vdash Y}{}
          }
        }
      }
    }
    &
    \infer[\orrtwo]{X \land Y \mid X \vdash (X \ot Y) \lor (Y \ot X)}{
      \infer[\andltwo]{X \land Y \mid X \vdash Y \ot X}{
        \infer[\tr]{Y \mid X \vdash Y \ot X}{
          \infer[\ax]{Y \mid \quad \vdash Y}{}
          &
          \infer[\pass]{- \mid X \vdash X}{
            \infer[\ax]{X \mid \quad \vdash X}{}
          }
        }
      }
    }
  }
  \end{array}
\end{displaymath}

We introduce a congruence relation $\circeq$ on the sets of cut-free derivations:
  \begin{equation}\label{fig:circeq}
  \small
  \begin{array}{rlll}
    \ax_{\I} &\circeq \unitl \text{ } (\unitr)
    \\
    \ax_{A \ot B} &\circeq \tl \text{ } (\tr \text{ } (\ax_{A} , \pass \text{ } \ax_{B}))
    \\
    \ax_{A \land B} &\circeq \andr \ (\andlone \ \ax_A , \andltwo \ \ax_B)
    \\
    \ax_{A \lor B} &\circeq \orl \ (\orrone \ \ax_A , \orrtwo \ \ax_B)
    \\
    \tr \text{ } (\pass \text{ } f, g) &\circeq \pass \text{ } (\tr \text{ } (f, g)) &&(f : A' \mid \Gamma \vdash A, g : {-} \mid \Delta \vdash B)
    \\
    \tr \text{ } (\unitl \text{ } f, g) &\circeq \unitl \text{ } (\tr \text{ } (f , g)) &&(f : {-} \mid \Gamma \vdash A , g : {-} \mid \Delta \vdash B)
    \\
    \tr \text{ } (\tl \text{ } f, g) &\circeq \tl \text{ } (\tr \text{ } (f , g)) &&(f : A' \mid B' , \Gamma \vdash A , g : {-} \mid \Delta \vdash B)
    \\
    \tr \ (\andli \ f , g) &\circeq \andli \ (\tr \ (f , g)) &&(f : A' \mid \Gamma \vdash A, g : - \mid \Delta \vdash B)
    \\
    \tr \ (\orl \ (f_1 , f_2) , g) &\circeq \orl \ (\tr \ (f_1 , g) , \tr \ (f_2 , g)) &&(f_1 : A' \mid \Gamma \vdash A , f_2 : B' \mid \Gamma \vdash A , \\
    & && \;\;g : - \mid \Delta \vdash B)
    \\
    \pass \ (\andr \ (f , g)) &\circeq  \andr \ (\pass \ f , \pass \ g) &&(f : A' \mid \Gamma \vdash A , g : A' \mid \Gamma \vdash B)
    \\
    \unitl \ (\andr \ (f , g)) &\circeq  \andr \ (\unitl \ f , \unitl \ g) &&(f : - \mid \Gamma \vdash A , g : - \mid \Gamma \vdash B)
    \\
    \tl \ (\andr \ (f , g)) &\circeq  \andr \ (\tl \ f , \tl \ g) &&(f : A' \mid B' , \Gamma \vdash A , g : A' \mid B' , \Gamma \vdash B)
    \\
    \andli \ (\andr \ (f , g)) &\circeq  \andr \ (\andli \ f , \andli \ g) &&(f : A' \mid \Gamma \vdash A , g : A' \mid \Gamma \vdash B)
    \\
    \orl \ (\andr \ (f_1 , g_1) , \andr \ (f_2 , g_2)) &\circeq  \andr \ (\orl \ (f_1 , f_2) , \orl \ (g_1 , g_2)) &&(f_1 : A' \mid \Gamma \vdash A , f_2 : B' \mid \Gamma \vdash A , 
    \\
    & &&\;\;g_1 : A' \mid \Gamma \vdash B , g_2 : B' \mid \Gamma \vdash B)
    \\
    \orri \ (\pass \ f) &\circeq \pass \ (\orri \ f) &&(f : A' \mid \Gamma \vdash A)
    \\
    \orri \ (\unitl \ f) &\circeq \unitl \ (\orri \ f) &&(f : - \mid \Gamma \vdash A)
    \\
    \orri \ (\tl \ f) &\circeq \tl \ (\orri \ f) &&(f : A' \mid B' , \Gamma \vdash A)
    \\
    \orri \ (\andli \ f) &\circeq \andli \ (\orri \ f) &&(f : A' \mid \Gamma \vdash A)
    \\
    \orri \ (\orl \ (f , g)) &\circeq \orl \ (\orri \ f , \orri \ g) &&(f : A' \mid \Gamma \vdash A , g : B' \mid \Gamma \vdash A)
  \end{array}
  \end{equation}
The first four equations ($\eta$-conversions) characterize the $\ax$ rule for non-atomic formulae. The remaining equations are permutative conversions. The congruence $\circeq$ has been carefully chosen to serve as the proof-theoretic counterpart of the equational theory of certain categories with skew structure, which we introduce in the next section.

\section{Categorical Semantics}\label{sec:categorical}
  A \emph{skew monoidal category} \cite{szlachanyi:skew-monoidal:2012,lack:skew:2012,lack:triangulations:2014} is a category $\mathbb{C}$ with a unit object $\I$, a functor $\ot : \mathbb{C} \times \mathbb{C} \rightarrow \mathbb{C}$
and three natural transformations $\lambda$, $\rho$, $\alpha$ typed
  $\lambda_A : \I \ot A \to A$, $\rho_A : A \to A \ot \I$ and $\alpha_{A,B,C} : (A \ot B) \ot C \to A \ot (B \ot C)$,
satisfying the following equations due to Mac Lane \cite{maclane1963natural}:
\begin{center}
\begin{tikzcd}
	& {\I \ot \I} \\[-.2cm]
	\I && \I
	\arrow["{\rho_{\I}}", from=2-1, to=1-2]
	\arrow["{\lambda_{\I}}", from=1-2, to=2-3]
	\arrow[Rightarrow, no head, from=2-1, to=2-3]
\end{tikzcd}
\qquad
\begin{tikzcd}
	{(A \ot \I) \ot B} & {A \ot (\I \ot B)} \\[-.3cm]
	{A \ot B} & {A \ot B}
	\arrow[Rightarrow, no head, from=2-1, to=2-2]
	\arrow["{\rho_A \ot B}", from=2-1, to=1-1]
	\arrow["{A \ot \lambda_{B}}", from=1-2, to=2-2]
	\arrow["{\alpha_{A , \I , B}}", from=1-1, to=1-2]
\end{tikzcd}

\begin{tikzcd}
	{(\I \ot A ) \ot B} && {\I \ot (A \ot B)} \\[-.3cm]
	& {A \ot B}
	\arrow["{\alpha_{\I , A ,B}}", from=1-1, to=1-3]
	\arrow["{\lambda_{A \ot B}}", from=1-3, to=2-2]
	\arrow["{\lambda_{A} \ot B}"', from=1-1, to=2-2]
\end{tikzcd}
\qquad
\begin{tikzcd}
	{(A \ot B) \ot \I} && {A \ot (B \ot \I)} \\[-.3cm]
	& {A \ot B}
	\arrow["{\alpha_{A , B, \I}}", from=1-1, to=1-3]
	\arrow["{A \ot \rho_B}"', from=2-2, to=1-3]
	\arrow["{\rho_{A \ot B}}", from=2-2, to=1-1]
\end{tikzcd}

\begin{tikzcd}
	{(A\ot (B\ot C)) \ot D} && {A \ot ((B \ot C) \ot D)} \\[-.2cm]
	{((A \ot B)\ot C) \ot D} & {(A \ot B) \ot (C \ot D)} & {A \ot (B \ot (C \ot D))}
	\arrow["{\alpha_{A , B\ot C , D}}", from=1-1, to=1-3]
	\arrow["{A \ot \alpha_{B , C ,D}}", from=1-3, to=2-3]
	\arrow["{\alpha_{A ,B ,C\ot D}}"', from=2-2, to=2-3]
	\arrow["{\alpha_{A \ot B , C , D}}"', from=2-1, to=2-2]
	\arrow["{\alpha_{A , B ,C} \ot D}", from=2-1, to=1-1]
\end{tikzcd}
\end{center}
A skew monoidal category with binary coproducts is \emph{(binary) left-distributive} if the canonical morphism typed $(A \ot C) + (B \ot C) \to (A + B) \ot C$ has an inverse $l : (A + B) \ot C \to (A \ot C) + (B \ot C)$. We will be interested in skew monoidal categories with binary products and coproducts, which moreover are left-distributive. We simply call these distributive skew monoidal categories.

A \emph{(strict) skew monoidal functor} $F : \mathbb{C} \rightarrow \mathbb{D}$ between skew monoidal categories $(\mathbb{C} , \I , \ot)$ and \linebreak$(\mathbb{D} , \I' , \ot')$ is a functor from $\mathbb{C}$ to $\mathbb{D}$ satisfying
    $F \I = \I'$ and $F (A \ot B) = F A \ot' F B$, also preserving the structural laws $\lambda$, $\rho$ and $\alpha$ on the nose. This means that $F \lambda_A^\mathbb{C} = \lambda_{FA}^\mathbb{D}$, where $\lambda^\mathbb{C}$ and $\lambda^\mathbb{D}$ are left-unitors of $\mathbb{C}$ and $\mathbb{D}$ respectively, and similar equations hold for $\rho$ and $\alpha$. A skew monoidal functor is \emph{distributive} if it also strictly preserves products, coproducts and (consequently also) left-distributivity.

The formulae, derivations and the equivalence relation $\circeq$ of the sequent calculus determine a \emph{syntactic} distributive skew monoidal category $\FSkMCC(\mathsf{At})$ (an acronym for \textsf{F}ree \textsf{D}istributive \textsf{Sk}ew \textsf{M}onoidal category on the set $\mathsf{At}$). Its objects are formulae. The operations $\I$ and $\ot$ are the logical connectives. The set of maps between objects $A$ and $B$ is the set of derivations $A \mid ~ \vdash B$ quotiented by the equivalence relation $\circeq$. The identity map on $A$ is the equivalence class of $\ax_A$, while composition is given by $\mathsf{scut}$. The structural laws $\lambda$, $\rho$, $\alpha$ are all admissible. Products and coproducts are the additive connectives $\land$ and $\lor$. Left-distributivity follows from the logical rules of $\lor$ and $\ot$.

Distributive skew monoidal categories form models of our sequent calculus.
Moreover the sequent calculus, as a presentation of a distributive skew monoidal category, is the \emph{initial} one among these
models. Equivalently, $\FSkMCC(\mathsf{At})$ is the \emph{free}
such category on the set $\mathsf{At}$.
\begin{theorem}\label{thm:models}
  Let $\mathbb{D}$ be a distributive skew monoidal category. Given a function $F_{\mathsf{At}} : \mathsf{At} \rightarrow |\mathbb{D}|$ evaluating atomic formulae as objects of $\mathbb{D}$, there exists a unique distributive skew monoidal functor \linebreak$F : \FSkMCC(\mathsf{At}) \rightarrow \mathbb{D}$ for which $F X = F_{\mathsf{At}} X$, for any atom $X$.
\end{theorem}
The construction of the functor $F$ and the proof of uniqueness proceed similarly to the proofs of Theorems 3.1 and 3.2 in \cite{UVW:protsn}.

\section{A Focused Sequent Calculus with Tag Annotations}\label{sec:focusing}
When oriented from left-to-right, the equations in (\ref{fig:circeq}) become a  rewrite system, which is locally confluent and strongly normalizing, thus confluent with unique normal forms. Here we provide an explicit description of the normal forms of (\ref{eq:seqcalc})  wrt.~this rewrite system.

For any sequent $S \mid \Gamma \vdash A$, a root-first proof search procedure can be defined as follows. First apply right invertible rules on the sequent until the principal connective of the succedent is non-negative, then apply left invertible rules until the stoup becomes either empty or non-positive. At this point, if we do not insist on focusing on a particular formula (either in the stoup or succedent, since no rule acts on formulae in context) as in Andreoli's focusing procedure~\cite{andreoli:logic:1992}, we obtain a sequent calculus with a reduced proof search space, that looks like this:
\begin{equation}\label{eq:naive:focus}
  \begin{array}{lc}
    \text{(right invertible)} &
    \proofbox{
      \infer[\andr]{S \mid \Gamma \vdash_{\RI} A \land B}{
        S \mid \Gamma \vdash_{\RI} A
        &
        S \mid \Gamma \vdash_{\RI} B
      }
    \qquad
    \infer[\LI 2 \RI]{S \mid \Gamma \vdash_{\RI} P}{S \mid \Gamma \vdash_{\LI} P}
    }
    \\[10pt]
    \text{(left invertible)} &
    \proofbox{
      \infer[\unitl]{\I \mid \Gamma \vdash_{\LI} P}{{-} \mid \Gamma \vdash_{\LI} P}
    \qquad
    \infer[\tl]{A \ot B \mid \Gamma \vdash_{\LI} P}{A \mid B , \Gamma \vdash_{\LI} P}
    \\
    \infer[\orl]{A \lor B \mid \Gamma \vdash_\LI P}{
      A \mid \Gamma \vdash_\LI P
      &
      B \mid \Gamma \vdash_\LI P
    }
    \qquad
    \infer[\F 2 \LI]{T \mid \Gamma \vdash_{\LI} P}{T \mid \Gamma \vdash_{\F} P}
    }
    \\[10pt]
    \text{(focusing)} &
    \proofbox{
    \infer[\pass]{{-} \mid A , \Gamma \vdash_{\F} P }{
        A \mid \Gamma \vdash_{\LI} P
    }
    \qquad
    \infer[\ax]{X \mid \quad \vdash_{\F} X}{}
    \qquad
    \infer[\unitr]{{-} \mid \quad \vdash_{\F} \I}{}
    }
    \\[10pt]
    \multicolumn{2}{c}{
    \infer[\tr]{T \mid \Gamma , \Delta \vdash_{\F} A \ot B}{
      T \mid \Gamma \vdash_{\RI} A
      &
      {-} \mid \Delta \vdash_{\RI} B
    }
    \qquad
    \infer[\orrone]{T \mid \Gamma \vdash_{\F} A \lor B}{
      T \mid \Gamma \vdash_{\RI} A
    }
    \qquad
    \infer[\orrtwo]{T \mid \Gamma \vdash_{\F} A \lor B}{
      T \mid \Gamma \vdash_{\RI} B
    }
    }
    \\[6pt]
    \multicolumn{2}{c}{
    \infer[\andlone]{A \land B \mid \Gamma \vdash_{\F} P}{A \mid \Gamma \vdash_{\LI} P}
    \qquad
    \infer[\andltwo]{A \land B \mid \Gamma \vdash_{\F} P}{B \mid \Gamma \vdash_{\LI} P}
    }
  \end{array}
\end{equation}
In the rules above, $P$ is a non-negative formula, i.e. its principal connective is not $\land$, and $T$ is a non-positive stoup (also called \emph{irreducible}), i.e. it is not $\I$ and its principal connective is neither $\ot$ nor $\lor$.

This calculus is too permissive. The same sequent $S \mid \Gamma \vdash_{\RI} A$ may have multiple derivations which correspond to $\circeq$-related derivations in the original sequent calculus. This happens since certain sequents in phase $\vdash_\F$ can be alternatively proved by an application of a left non-invertible rule ($\pass$, $\andlone$ and $\andltwo$) or an application of a right non-invertible rule ($\tr$, $\orrone$ and $\orrtwo$). As concrete  examples, both sequents $- \mid X , Y \vdash_{\F} X \ot Y$ and $X \land Y \mid ~ \vdash_{\F} X \lor Y$ have multiple distinct proofs in this calculus, but their corresponding proofs in the original calculus are $\circeq$-related.

In phase $\vdash_\F$, only non-invertible rules can be applied, so the question is: how to arrange the order between non-invertible rules without causing undesired non-determinism and losing completeness with respect to the sequent calculus in (\ref{eq:seqcalc}) and its equivalence relation $\circeq$? Similarly to \cite{UVW:protsn}, our strategy is to prioritize left non-invertible rules over right ones, unless this does not lead to a valid derivation and the other way around is necessary.
For example, consider the sequent $X \land Y \mid \quad \vdash_{\F} (X \land Y) \lor Z$. Proof search fails if we apply $\andli$ before $\orrone$. A valid proof is obtained only when applying $\orrone$ before $\andli$. Rule $\sw$ is an abbreviation for the application of multiple consecutive phase switching rules.

\begin{equation}\label{eq:ex:focused}
  \small
  \begin{array}{c}
     \infer[\orrone]{X \land Y \mid \quad \vdash_\F (X \land Y) \lor Z}{
      \infer[\andr]{X \land Y \mid \quad \vdash_\RI X \land Y}{
        \infer[\sw]{X \land Y \mid \quad \vdash_\RI X}{
          \infer[\andlone]{X \land Y \mid \quad \vdash_\F X}{
          \infer[\F 2 \LI]{X \mid \quad \vdash_\LI X}{
          \infer[\ax]{X \mid \quad \vdash_\F X}{}
        }}}
        &
        \infer[\sw]{X \land Y \mid \quad \vdash_\RI Y}{
        \infer[\andltwo]{X \land Y \mid \quad \vdash_\F Y}{
          \infer[\F 2 \LI]{Y \mid \quad \vdash_\LI Y}{
          \infer[\ax]{Y \mid \quad \vdash_\F Y}{}
        }}}
      }
    }
  \end{array}
\end{equation}
In this example it was possible to first apply $\orrone$ since, after the application of $\andr$, different left \linebreak$\land$-rules are applied in different branches of the proof tree. If we would have applied  the same rule $\andlone$ to both premises (imagine that $X = Y$ for this to be possible), then we could have obtained a $\circeq$-equivalent derivation by moving the application of $\andlone$ to the bottom of the proof tree. 

It is also possible that the two premises of $\andr$ correspond to $\circeq$-inequivalent derivations.
For example, consider the following proof of sequent $- \mid \I \ot X \vdash_\F ((I \ot X) \land (\I \ot X)) \lor Y$:
\begin{equation*}\small
  \begin{array}{c}
    \infer[\orrone]{- \mid \I \ot X \vdash_\F ((\I \ot X) \land (\I \ot X)) \lor Y}{
      \infer[\andr]{- \mid \I \ot X \vdash_\RI (\I \ot X) \land (\I \ot X)}{
        \infer[\sw]{- \mid \I \ot X \vdash_\RI \I \ot X}{
        \infer[\pass]{- \mid \I \ot X \vdash_\F \I \ot X}{
          \infer[\tl]{\I \ot X \mid \quad \vdash_\LI \I \ot X}{
          \infer[\unitl]{\I \mid X \vdash_\LI \I \ot X}{
            \infer[\F 2 \LI]{- \mid X \vdash_\LI \I \ot X}{
            \infer[\tr]{- \mid X \vdash_\F \I \ot X}{
                \infer[\sw]{- \mid \quad \vdash_\RI \I}{
                \infer[\unitr]{- \mid \quad \vdash_\F \I}{}
              }
              &
              \infer[\sw]{- \mid X \vdash_\RI X}{
              \infer[\pass]{- \mid X \vdash_\F X}{
                \infer[\F 2 \LI]{X \mid \quad \vdash_\LI X}{
                \infer[\ax]{X \mid \quad \vdash_\F X}{}
              }
            }
          }
        }
        }}}}}
        &
        \infer[\sw]{- \mid \I \ot X \vdash_\RI \I \ot X}{
        \infer[\tr]{- \mid \I \ot X \vdash_\F \I \ot X}{
          \infer[\sw]{- \mid \quad \vdash_\RI \I}{
            \infer[\unitr]{- \mid \quad \vdash_\F \I}{}
            }
          &
          \infer[\sw]{- \mid \I \ot X \vdash_\RI X}{
          \infer[\pass]{- \mid \I \ot X \vdash_\F X}{
            \infer[\tl]{\I \ot X \mid \quad \vdash_\LI X}{
              \infer[\unitl]{\I \mid X \vdash_\LI X}{
                \infer[\F 2 \LI]{- \mid X \vdash_\LI X}{
                \infer[\pass]{- \mid X \vdash_\F X}{
                  \infer[\F 2 \LI]{X \mid \quad \vdash_\LI X}{
                  \infer[\ax]{X \mid \quad \vdash_\F X}{}
                }
              }
            }
          }
        }
      }
    }}}}}
  \end{array}
\end{equation*}
In this case,
the right non-invertible rule $\orrone$ must be applied before the left non-invertible rule $\pass$. This is because $\pass$ is used in the proof of the left branch of $\andr$, but it is not used in the proof of the right branch, $\tr$ is used instead. If both left and right proofs would have used $\pass$ (for example, they could have been the same exact proof), then it would have been possible to apply $\pass$ before $\orrone$.

In general, a right non-invertible rule should be applied  before a left non-invertible one if, after the possible application of some $\andr$ rules, either:
$(i)$ a right non-invertible rule or the $\ax$ rule is applied to one of the premises;
$(ii)$ $\andlone$ and $\andltwo$ are applied to different premises.
Therefore, we have to make sure that in the focused sequent calculus, after the application of a right non-invertible rule, not all premises use the same left non-invertible rule, because in this case the latter rule could be applied first.

In order to keep track of this,
we use a system of \emph{tag annotations} and we introduce new phases of proof search where sequents are annotated by \emph{list of tags}. 
There are four tags: $\tP , \tCone , \tCtwo , \tE$. Intuitively, they respectively correspond to rules $\pass , \andlone , \andltwo$ and all the remaining non-invertible rules in phase $\vdash_\F$.
A list of tags $l$ is called \emph{valid} if it is non-empty and either $(i)$ $\tE \in l$ or $(ii)$ both $\tCone \in l$ and $\tCtwo \in l$.

Derivations in the focused sequent calculus with tag annotations are generated by the rules
\begin{equation}\label{eq:focus}
  \begin{array}{lc}
    \text{(right invertible)} &
    \proofbox{
      \infer[\andr]{S \mid \Gamma \vdash^{l_1?, l_2?}_{\RI} A \land B}{
        S \mid \Gamma \vdash^{l_1?}_{\RI} A
        &
        S \mid \Gamma \vdash^{l_2?}_{\RI} B
      }
    \qquad
    \infer[\LI 2 \RI]{S \mid \Gamma \vdash^{t?}_{\RI} P}{S \mid \Gamma \vdash^{t?}_{\LI} P}
    }
    \\[10pt]
    \text{(left invertible)} &
    \proofbox{
      \infer[\unitl]{\I \mid \Gamma \vdash_{\LI} P}{{-} \mid \Gamma \vdash_{\LI} P}
    \qquad
    \infer[\tl]{A \ot B \mid \Gamma \vdash_{\LI} P}{A \mid B , \Gamma \vdash_{\LI} P}
    \\
    \infer[\orl]{A \lor B \mid \Gamma \vdash_\LI P}{
      A \mid \Gamma \vdash_\LI P
      &
      B \mid \Gamma \vdash_\LI P
    }
    \qquad
    \infer[\F 2 \LI]{T \mid \Gamma \vdash^{t?}_{\LI} P}{T \mid \Gamma \vdash^{t?}_{\F} P}
    }
    \\[10pt]
    \text{(focusing)} &
    \proofbox{
    \infer[\pass]{{-} \mid A , \Gamma \vdash^{\tP?}_{\F} P }{
        A \mid \Gamma \vdash_{\LI} P
    }
    \qquad
    \infer[\ax]{X \mid \quad \vdash^{\tE?}_{\F} X}{}
    \qquad
    \infer[\unitr]{{-} \mid \quad \vdash^{\tE?}_{\F} \I}{}
    }
    \\[10pt]
    \multicolumn{2}{c}{
    \infer[\tr]{T \mid \Gamma , \Delta \vdash^{\tE?}_{\F} A \ot B}{
      T \mid \Gamma \vdash^{l}_{\RI} A
      &
      {-} \mid \Delta \vdash_{\RI} B
      &
      l \ \text{valid}
    }
    \quad
    \infer[\orrone]{T \mid \Gamma \vdash^{\tE?}_{\F} A \lor B}{
      T \mid \Gamma \vdash^{l}_{\RI} A
      &
      l \ \text{valid}
    }
    \quad
    \infer[\orrtwo]{T \mid \Gamma \vdash^{\tE?}_{\F} A \lor B}{
      T \mid \Gamma \vdash^{l}_{\RI} B
      &
      l \ \text{valid}
    }
    }
    \\[6pt]
    \multicolumn{2}{c}{
    \infer[\andlone]{A \land B \mid \Gamma \vdash^{\tCone?}_{\F} P}{A \mid \Gamma \vdash_{\LI} P}
    \qquad
    \infer[\andltwo]{A \land B \mid \Gamma \vdash^{\tCtwo?}_{\F} P}{B \mid \Gamma \vdash_{\LI} P}
    }
  \end{array}
\end{equation}
We use $l$ for lists of tags and $t$ for single tags.
The notation $l?$ indicates that the sequent is either untagged or assigned the list of tags $l$. Similarly for notation $t?$.
We discuss the proof search procedures of untagged and tagged sequents separately.
The proof search of a sequent $S \mid \Gamma \vdash_\RI A$ proceeds as follows:
\begin{itemize}
  \item[($\vdash_{\RI}$)] We apply the right invertible rule $\andr$ eagerly to decompose the succedent until its principal connective is not $\land$, then we move to the left invertible phase $\vdash_\LI$ with an application of $\LI 2 \RI$.
  \item[($\vdash_{\LI}$)] We apply left invertible rules until the stoup becomes irreducible, then move to the focusing phase $\vdash_\F$ with an application of $\F 2 \LI$.
  \item[($\vdash_{\F}$)] We apply one of the remaining rules. Since the sequents are not marked by tags at this point, rules $\pass$, $\ax$ , $\unitr$ and $\andli$ can be directly applied when stoups, contexts and succedents are of the appropriate form.
  If we decide to apply a right non-invertible rule, we need to come up with a valid list of tags $l$ and subsequently continue proof search in tagged right invertible phase $\vdash^{l}_{\RI}$, which is described below. Notice that only the first premise of $\tr$ is tagged, the second premise is not, i.e. its proof search continues in phase $\vdash_{\RI}$.
\end{itemize}
The proof search of a sequent $T \mid \Gamma \vdash^{l}_\RI A$ proceeds as follows (notice that at this point in proof search the stoup $T$ is necessarily irreducible):
\begin{itemize}
  \item[($\vdash^{l}_{\RI}$)] We apply the $\andr$ rule to decompose the succedent and split the list of tags carefully until the succedent becomes non-negative and the list of tags becomes a singleton $t$, then we move to phase $\vdash_\LI^t$ with an application of $\LI2 \RI$.
  \item[($\vdash^{t}_{\LI}$)] Since the stoup is either empty or a negative formula, we immediately switch to phase $\vdash_\F$ with an application of $\F 2 \LI$. This motivates why sequents in rules $\unitl$, $\tl$, $\orl$ are not tagged.
  \item[($\vdash^{t}_{\F}$)] If $t = \tE$ we can apply either $\ax$, $\unitr$ or another right non-invertible rule. Again, when applying right non-invertible rules we need to come up with a new valid list of tags.
  Left non-invertible rules can be applied only when the tag is correct, i.e. $\pass$ with tag $\tP$, $\andlone$ with tag $\tCone$, and $\andltwo$ with tag $\tCtwo$.
\end{itemize}
The derivation in (\ref{eq:ex:focused}) can be reconstructed in the focused calculus with tag annotations.
\begin{equation}\label{eq:ex:tag:focused}
  \begin{array}{c}
  \infer[\orrone]{X \land Y \mid \quad \vdash_{\F} (X \land Y) \lor Z}{
    \infer[\andr]{X \land Y \mid \quad \vdash^{\tCone , \tCtwo}_{\RI} X \land Y}{
      \infer[\sw]{X \land Y \mid \quad \vdash^{\tCone}_{\RI} X}{
        \infer[\andlone]{X \land Y \mid \quad \vdash^{\tCone}_{\F} X}{
          \infer[\sw]{X \mid \quad \vdash_{\LI} X}{
            \infer[\ax]{X \mid \quad \vdash_{\F} X}{}
          }
        }
      }
      &
      \infer[\sw]{X \land Y \mid \quad \vdash^{\tCtwo}_{\RI} Y}{
        \infer[\andltwo]{X \land Y \mid \quad \vdash^{\tCtwo}_{\F} Y}{
          \infer[\sw]{Y \mid \quad \vdash_{\LI} Y}{
            \infer[\ax]{Y \mid \quad \vdash_{\F} Y}{}
          }
        }
      }
    }
  }
  \end{array}
\end{equation}
Notice that the list of tags is not predetermined when a right non-invertible rule is applied, we have to come up with one ourselves.
Practically, the list $l$ can be computed by continuing proof search until, in each branch, we hit the first application of a rule in phase $\vdash_\F$, each with its own (necessarily uniquely determined) single tag $t$. Take $l$ as the concatenation of the resulting $t$s and check whether it is valid. If it is not, backtrack and apply a left non-invertible rule instead.
\begin{theorem}\label{theorem:focus:sound:complete}
  The focused sequent calculus with tag annotations in (\ref{eq:focus}) is sound and complete with respect to the sequent calculus in (\ref{eq:seqcalc}).
\end{theorem}
Soundness is immediate because there exist functions $\mathsf{emb}_{ph} : S \mid \Gamma \vdash^{l?}_{ph} A \to S \mid \Gamma \vdash A$, for all $ph \in \{ \RI , \LI , \F \}$, which erase all phase and tag annotations.
Completeness follows from the fact that the following rules are all admissible:
\begin{equation}\label{eq:admis}
  \small
  \hspace*{-4mm}
    \begin{array}{c}
      \infer[\unitl^{\RI}]{\I \mid \Gamma \vdash_{\RI} C}{{-} \mid \Gamma \vdash_{\RI} C}
      \quad
      \infer[\tl^{\RI}]{A \ot B \mid \Gamma \vdash_{\RI} C}{A \mid B, \Gamma \vdash_{\RI} C}
      \quad
      \infer[\pass^{\RI}]{{-} \mid A , \Gamma \vdash_{\RI} C}{A \mid \Gamma \vdash_{\RI} C}
      \quad
      \infer[\ax^{\RI}]{A \mid \quad \vdash_{\RI} A}{}
      \quad
      \infer[\unitr^{\RI}]{{-} \mid \quad \vdash_{\RI} \I}{}
  \\[6pt]
      \infer[\orl^{\RI}]{A \lor B \mid \Gamma \vdash_{\RI} C}{
      A \mid \Gamma \vdash_{\RI} C
      &
      B \mid \Gamma \vdash_{\RI} C
      }
      \qquad
      \infer[\tr^{\RI}]{S \mid \Gamma , \Delta \vdash_{\RI} A \ot B}{
        S \mid \Gamma \vdash_{\RI} A
        &
        {-} \mid \Delta \vdash_{\RI} B
      }
  \\[6pt]
      \infer[\andlone^{\RI}]{A \land B \mid \Gamma \vdash_{\RI} C}{A \mid \Gamma \vdash_{\RI} C}
      \qquad
      \infer[\andltwo^{\RI}]{A \land B \mid \Gamma \vdash_{\RI} C}{B \mid \Gamma \vdash_{\RI} C}
      \qquad
      \infer[\orrone^{\RI}]{S \mid \Gamma \vdash_{\RI} A \lor B}{S \mid \Gamma \vdash_{\RI} A}
      \qquad
      \infer[\orrtwo^{\RI}]{S \mid \Gamma \vdash_{\RI} A \lor B}{S \mid \Gamma \vdash_{\RI} B}
    \end{array}
  \end{equation}
The admissibility of the rules in (\ref{eq:admis}), apart from the right non-invertible ones, is proved by structural induction on derivations.
The same strategy cannot be applied to right non-invertible rules.
For example, if the premise of $\orrone^{\RI}$ ends with an application of $\andr$, we get immediately stuck:
\begin{displaymath}
  \begin{array}{c}
    \small
    \proofbox{
    \infer[\orrone^{\RI}]{S \mid \Gamma \vdash_{\RI} (A' \land B') \lor B}{
      \infer[\andr]{S \mid \Gamma \vdash_{\RI} A' \land B'}{
        \deduce{S \mid \Gamma \vdash_{\RI} A'}{f}
        &
        \deduce{S \mid \Gamma \vdash_{\RI} B'}{g}
      }
    }
    }
    \quad
    =
    \quad
    ??
  \end{array}
\end{displaymath}
The inductive hypothesis applied to $f$ and $g$ would produce wrong sequents for the target conclusion.
This is fixed by proving the admissibility of more general rules. In order to state and prove this, we need to first introduce a few lemmata. 
The first one shows that applying several $\andr$ rules in one step is admissible.

  Let $\conj{A}$ be the list of formulae obtained by decomposing additive conjunctions $\land$ in the formula $A$. Concretely, $\conj{A} = \conj{A'} , \conj{B'}$ if $A = A' \land B'$ and $\conj{A} = A$ otherwise.
\begin{lemma}\label{lem:BigStep}
  The following rules
  \begin{displaymath}
    \infer[\andr^{*}_t]{T \mid \Gamma \vdash^{l}_{\RI} A}{
      [T \mid \Gamma \vdash^{t_i}_{\F} P_i]_{i \in [1 , \dots , n]}
    }
    \qquad
    \infer[\andr^{*}]{T \mid \Gamma \vdash_{\RI} A}{
      [T \mid \Gamma \vdash_{\LI} P_i]_{i \in [1 , \dots , n]}
    }
  \end{displaymath}
  are admissible, where $\conj{A} = [P_1 , \dots , P_n]$ and $l = [t_1 , \dots , t_n]$.
\end{lemma}
\begin{proof}
  We show the case of $\andr^{*}_{t}$, the other one is similar.
  Let $fs : [T \mid \Gamma \vdash^{t_i}_{\F} P_i]_i$ be a list of derivations. The proof proceeds by induction on $A$.
  \begin{itemize}
    \item If $A \neq A' \land B'$, then $fs$ consists of a single derivation $f$. Define $\andr^*_{t} \text{ } fs = \F 2 \LI \text{ } (\LI 2 \RI \text{ } f)$.
    \item If $A = A' \land B'$, then there exist lists of derivations $fs_1 : [T \mid \Gamma \vdash^{t_i}_{\F} A'_i]_{i \in [1 , \dots , m]}$ and \linebreak 
    $fs_2 : [T \mid \Gamma \vdash^{t_i}_{\F} B'_i]_{i \in [m+1 , \dots , n]}$, and lists of tags $l_1 = t_1 , \dots , t_m$ and $l_2 = t_{m+1} , \dots , t_n$, so that $fs$ is the concatenation of $fs_1$ and $fs_2$ and $l$ is the concatenation of $l_1 $ and $l_2$.
    Apply $\andr$ at the bottom, then proceed recursively:
    \begin{displaymath}
      \small
      \begin{array}{cc}
        \proofbox{
        \infer[\andr^*_{t}]{T \mid \Gamma \vdash^{l} A' \land B'}{
          \deduce{[T \mid \Gamma \vdash^{t_i}_{\F} P_i]_{i \in [1 , \dots , n]}}{fs}
        }
        }
        &
        =
        \proofbox{
        \infer[\andr]{T \mid \Gamma \vdash^{l_1 , l_2}_{\RI} A' \land B'}{
          \infer[\andr^{*}_{t}]{T \mid \Gamma \vdash^{l_1}_{\RI} A'}{
            \deduce{[T \mid \Gamma \vdash^{t_i}_{\F} P_i]_{i \in [1 , \dots , m]}}{fs_1}
          }
          &
          \infer[\andr^{*}_{t}]{T \mid \Gamma \vdash^{l_2}_{\RI} B'}{
            \deduce{[T \mid \Gamma \vdash^{t_i}_{\F} P_i]_{i \in [m+1 , \dots , n]}}{fs_2}
          }
        }
       }
      \end{array}
    \end{displaymath}
  \end{itemize}
\end{proof}

The second lemma corresponds to the invertibility of phase $\vdash_\RI$.

\begin{lemma}\label{lem:RI:invert}
  Given $f : S \mid \Gamma \vdash_{\RI} A$, there is a list of derivations $fs : [S \mid \Gamma \vdash_{\LI} P_i]_{i \in [1 , \dots , n]}$ with $f = \andr^{*} fs$.
\end{lemma}
\begin{proof}
  The proof proceeds by structural induction on $f : S \mid \Gamma \vdash_{\RI} A$.
  \begin{itemize}
    \item If $f = \LI2 \RI \text{ } f_1$, then $A$ is non-negative. Take $fs$ as the singleton list consisting exclusively of $f_1$.
    \item If $f = \andr \ (f_1, f_2)$, then by inductive hypothesis we have $fs_1 : [S \mid \Gamma \vdash_{\LI} P_i]_{i \in [1 , \dots , n]}$ and \linebreak
    $fs_2 : [S \mid \Gamma \vdash_{\LI} P'_i]_{i \in [1 , \dots , m]}$. Take $fs$ as the concatenation of $fs_1$ and $fs_2$.
  \end{itemize}
\end{proof}

\begin{proposition}\label{prop:GenRightRules}
  The following rules
  \begin{displaymath}
    \begin{array}{c}
      \infer[\orrone^{\LI}]{S \mid \Gamma \vdash_{\LI} A \lor B}{
        \deduce{[S \mid \Gamma \vdash_{\LI} P_i]_{i \in [1 , \dots , n]}}{fs}
      }
      \quad
      \infer[\orrtwo^{\LI}]{S \mid \Gamma \vdash_{\LI} A \lor B}{
        \deduce{[S \mid \Gamma \vdash_{\LI} Q_i]_{i \in [1 , \dots , m]}}{fs}
      }
      \quad
      \infer[\tr^{\LI}]{S \mid \Gamma , \Delta \vdash_{\LI} A \ot B'}{
        \deduce{[S \mid \Gamma \vdash_{\LI} P_i]_{i \in [1 , \dots , n]}}{fs}
        &
        - \mid \Delta \vdash_{\RI} B'
      }
    \end{array}
  \end{displaymath}
  are admissible, where $\conj{A} = [P_1 , \dots , P_n]$ and $\conj{B} = [Q_1 , \dots , Q_m]$.
\end{proposition}
\begin{proof}
  The list of derivations $fs$ is non-empty, so we let $fs = [f_1 , fs']$.
  We proceed by induction on $f_1$.
  We only present the proof for $\orrone^{\LI}$, the admissibility of $\orrtwo^{\LI}$ and $\tr^{\LI}$ is proved similarly.
  
  If $f_1$ ends with the application of a left invertible rule, then all the derivations in $fs'$ necessarily end with the same rule as well.
  Therefore, we permute this rule with $\orrone^{\LI}$ and apply the inductive hypothesis.

  If $f_1 = \F 2 \LI \,f'_1$, then all the derivations in $fs'$ necessarily end with $\F 2 \LI$ as well. 
    We generate a list of tags $l$ by examining the shape of each derivation in $fs$: we add $\tP$ for each $\pass$, $\tCone$ for each $\andlone$, $\tCtwo$ for each $\andltwo$ and $\tE$ for the remaining rules.
    There are two possibilities:
    \begin{itemize}
      \item The resulting list $l$ is valid. We switch to phase $\vdash_\F$ and apply  $\orrone^{\LI}$ followed by $\andr^{*}_{t}$:
      \begin{displaymath}
        \begin{array}{cc}
          \proofbox{
          \infer[\orrone^{\LI}]{T \mid \Gamma \vdash_{\LI} A \lor B}{
            \infer[ {[\F 2 \LI]} ]{[T \mid \Gamma \vdash_{\LI} P_i]_{i \in [1 , \dots , n]}}{
              \deduce{[T \mid \Gamma \vdash_{\F} P_i]_{i \in [1 , \dots , n]}}{fs^{*}}
            }
          }
          }
          &
          =
          \proofbox{
          \infer[\F 2 \LI]{T \mid \Gamma \vdash_{\LI} A \lor B}{
            \infer[\orrone]{T \mid \Gamma \vdash_{\F} A \lor B}{
              \infer[\andr^{*}_t]{T \mid \Gamma \vdash^{l}_{\RI} A}{
                \deduce{[T \mid \Gamma \vdash^{t_i}_{\F} P_i]_{i \in [1 , \dots , n]}}{{fs^{*}}'}
              }
            }
          }
        }
        \end{array}
      \end{displaymath}
      A rule wrapped in square brackets, like $[\F 2 \LI]$ above, denotes the application of the rule to the conclusion of each derivation in the list.
      The list of derivations $fs^{*}$ is obtained from $fs$ by applying $[\F 2 \LI]$, i.e. $fs = [\F 2 \LI]\,fs^*$, while ${fs^*}'$ is a list of derivations whose conclusions are tagged version of those in $fs^{*}$, which can be easily constructed from $fs^{*}$.
      \item The list $l$ is invalid. In this case, all elements in $fs$ end with the same left non-invertible rule, so we permute the rule down with $\orrone^{\LI}$ and continue recursively.
      Here is an example where all derivations in $fs$ end with an application of $\pass$, i.e. $fs = [\F 2 \LI] \,([\pass]\,fs^*)$:
      \begin{displaymath}
        \begin{array}{cc}
          \proofbox{
            \infer[\orrone^{\LI}]{- \mid A' , \Gamma \vdash_{\LI} A \lor B}{
              \infer[ {[\F 2 \LI]} ]{[- \mid A' , \Gamma \vdash_{\LI} P_i]_{i \in [1 , \dots , n]}}{
                \infer[ {[\pass]} ]{[- \mid A' , \Gamma \vdash_{\LI} P_i]_{i \in [1 , \dots , n]}}{
                  \deduce{[A' \mid \Gamma \vdash_{\LI} P_i]_{i \in [1 , \dots , n]}}{fs^{*}}
                }
              }
            }
          }
          &
          =
          \proofbox{
            \infer[\F 2 \LI]{- \mid A' , \Gamma \vdash_{\LI} A \lor B}{
              \infer[\pass]{- \mid A' , \Gamma \vdash_{\F} A \lor B}{
                \infer[\orrone^{\LI}]{A' \mid \Gamma \vdash_{\LI} A \lor B}{
                  \deduce{[A' \mid \Gamma \vdash_{\LI} P_i]_{i \in [1 , \dots , n]}}{fs^{*}}
                }
              }
            }
          }
        \end{array}
      \end{displaymath}
    \end{itemize}
\end{proof}
Finally, a right non-invertible rule in (\ref{eq:admis}) is defined as follows: first invert its premises (for $\tr^{\RI}$, only the left premise) using Lemma \ref{lem:RI:invert}. Then apply the corresponding generalized rule in Proposition \ref{prop:GenRightRules}.

We can construct a function $\mathsf{focus} : S \mid \Gamma \vdash A \to S \mid \Gamma \vdash_{\RI} A$ by structural recursion on the input derivation. Each inference rule in (\ref{eq:seqcalc}) is sent to the corresponding admissible rule in (\ref{eq:admis}).
For example, $\mathsf{focus} \ (\orrone \ f) = \orrone^{\RI} \ (\mathsf{focus} \ f)$. Furthermore, it can be shown that $\mathsf{emb}_{\RI}$ and $\mathsf{focus}$ are each other inverses, in the sense made precise by the following theorem.
\begin{theorem}
  The functions $\mathsf{emb}_{\RI}$ and $\mathsf{focus}$ define a bijective correspondence between the set of derivations of $S \mid \Gamma \vdash A$ quotiented by the equivalence relation $\circeq$ and the set of derivations of $S \mid \Gamma \vdash_{\RI} A$: 
  \begin{itemize}
    \item For all $f, g : S \mid \Gamma \vdash A$, if $f \circeq g$ then $\mathsf{focus} \text{ } f = \mathsf{focus} \text{ } g$.
    \item For all $f : S \mid \Gamma \vdash A$, $\mathsf{emb}_{\RI} \text{ } (\mathsf{focus} \text{ } f) \circeq f$.
    \item For all $f : S \mid \Gamma \vdash_{\RI} A$, $\mathsf{focus} \text{ } (\mathsf{emb}_{\RI} \text{ } f) = f$.
  \end{itemize}
\end{theorem}
\begin{proof}
  The first bullet is proved by structural induction on the given equality proof $e : f \circeq g$. The other bullets are proved by structural induction on $f$. See the associated \href{https://github.com/cswphilo/SkewMonAdd/blob/main/skew-mon-conjunction-disjunction/Main.agda}{Agda formalization} for details.
\end{proof}

\section{Extensions of the Logic}\label{sec:extensions}
We now discuss some extensions of the sequent calculus and the focusing strategy.

\subsection{Additive Units}\label{subsec:AddUnits}

The sequent calculus in (\ref{eq:seqcalc}) can be made ``fully'' additive by including two units $\top$ and $\bot$ (the latter named 0 in linear logic literature),
two new introduction rules and two new generating equations:

\begin{displaymath}
  \begin{array}{c}
    \infer[\topr]{S \mid \Gamma \vdash \top}{}
    \qquad
    \infer[\botl]{\bot \mid \Gamma \vdash C}{}
  \end{array}
\end{displaymath}
\begin{equation*}\label{eq:units}
  \small
  \begin{array}{rclll}
    f & \circeq \topr && (f : S \mid \Gamma \vdash \top) \\
    f & \circeq \botl && (f : \bot \mid \Gamma \vdash C)
  \end{array}
\end{equation*}
In the focused sequent calculus we add $\topr$ in phase $\vdash_\RI$ and $\botl$ in phase $\vdash_\LI$, so that they can be applied as early as possible.
We include a new tag $\tT$ for $\top$. The validity condition for lists of tags is updated as follows: ($i$) $\tT$ or $\tE \in l$ or ($ii$) both $\tCone \in l$ and $\tCtwo \in l$. 
\begin{equation*}\label{eq:focus:units}
  \begin{array}{c}
    \infer[\topr]{S \mid \Gamma \vdash^{\tT?}_{\RI} \top}{}
    \qquad
    \infer[\botl]{\bot \mid \Gamma \vdash_{\LI} P}{}
  \end{array}
\end{equation*}

Categorical models of the extended sequent calculus are the distributive monoidal categories of Section~\ref{sec:categorical} with additionally a terminal and an initial object, which moreover satisfy a \emph{(nullary) left-distributivity} (or \emph{absorption}) condition: the canonical morphism typed $0 \to 0 \ot C$ has an inverse\linebreak $k : 0 \ot C \to 0$. The latter is used in the interpretation of the rule $\botl$.

\subsection{Skew Exchange}\label{subsec:Ex}

Following \cite{veltri:coherence:2021}, we consider a ``skew'' commutative extension of the sequent calculus in (\ref{eq:seqcalc}) obtained by adding a rule swapping adjacent  formulae in context:
\begin{displaymath}
  \infer[\ex]{S \mid \Gamma , B , A , \Delta \vdash C}{
    S \mid \Gamma , A , B , \Delta \vdash C
  }
\end{displaymath}
Note that exchanging the formula in the stoup, whenever the latter is non-empty, with a formula in context is not allowed. The new rule $\ex$ comes with additional generating equations for the congruence relation $\circeq$:
\begin{equation*}\label{fig:circeq:sym}
  \small
  \arraycolsep=2pt
  \def\arraystretch{1.1}
  \hspace{-8.9pt}
\begin{array}{rclll}
\ex_{B , A}  (\ex_{A , B}  f) &\circeq & f &&(f: S \mid \Gamma , A , B , \Delta \vdash C)
\\
\ex_{A , B}  (\ex_{A , D}  (\ex_{B , D}  f)) &\circeq & \ex_{B , D}  (\ex_{A , D}  (\ex_{A , B}  f)) &&(f : S \mid \Gamma , A , B , D , \Delta \vdash C) 
\\
  \andli \ (\ex_{A , B} \ f) &\circeq &\ex_{A , B} \ (\andli \ f) &&(f : A' \mid \Gamma , A , B , \Delta \vdash C)
  \\
  \andr \ (\ex_{A , B} \ f , \ex_{A , B} \ g) &\circeq &\ex_{A , B} \ (\andr \ (f , g)) &&(f : S \mid \Gamma , A , B , \Delta \vdash A' , g : S \mid \Gamma , A , B , \Delta \vdash B')
  \\
  \orl \ (\ex_{A , B} \ f , \ex_{A , B} \ g) &\circeq &\ex_{A , B} \ (\orl \ (f , g)) &&(f : A' \mid \Gamma , A , B , \Delta \vdash C ,  g : B' \mid \Gamma , A , B , \Delta \vdash C)
  \\
  \orri \ (\ex_{A , B} \ f) &\circeq &\ex_{A , B} \ (\orri \ f) &&(f : S \mid \Gamma , A , B \vdash A')
  \\
  \ex_{A , B}  (\ex_{A' , B'}  f) &\circeq & \ex_{A' , B'}  (\ex_{A , B}  f) &&(f: S \mid \Gamma , A , B , \Delta , A' , B' , \Lambda \vdash C)
\end{array}
\end{equation*}
The first equation states that swapping the same two formulae twice yields the same result as doing nothing.
The second equation corresponds to the Yang-Baxter equation.
The remaining equations are permutative conversions.
We left out permutative conversions describing the relationship between $\ex$ and the rules $\pass$, $\unitl$, $\tl$ and $\tr$, which can be found in \cite[Fig. 2]{veltri:coherence:2021}.

The resulting sequent calculus enjoys categorical semantics in distributive skew \emph{symmetric} monoidal categories, that possess a natural isomorphism $s_{A , B , C} : A \ot (B \ot C) \to A \ot (C \ot B)$ representing a form of ``skew symmetry''  involving three objects instead of two \cite{bourke:lack:braided:2020}.

The focused sequent calculus is extended with a new phase $\vdash_\C$ (for `context`) where the exchange rule can be applied. Rule $\tl$ has to be modified, since we need to give the possibility to move the formula $B$ to a different position in the context.
\begin{equation*}\label{eq:focus:sym}
  \begin{array}{c}
      \infer[\ex]{S \mid \Gamma , A \spl \Delta , \Lambda \vdash_{\C} C}{S \mid \Gamma \spl \Delta , A , \Lambda \vdash_{\C} C}
      \qquad
      \infer[\RI 2 \C]{S \mid \quad \spl \Gamma \vdash_{\C} C}{S \mid \Gamma \vdash_{\RI} C}
    \qquad
    \infer[\tl]{A \ot B \mid \Gamma \vdash_{\LI} P}{A \mid B \spl \Gamma \vdash_{\C} P}
  \end{array}
\end{equation*}
 Root-first proof search now begins in the new phase $\vdash_\C$, where formulae in context are permuted.
  We start with a sequent $S \mid \Gamma \spl \quad \vdash_{\C}C$ and end with a sequent $S \mid \quad \spl \Gamma' \vdash_{\C}C$ where $\Gamma'$ is a permutation of $\Gamma$.
  In the process, the context is divided into two parts $\Gamma \spl \Delta$, where the formulae in $\Gamma$ are ready to be moved while those in $\Delta$ have already been placed in their final position.
  Once all formulae in $\Gamma$ have been moved, we switch to phase $\vdash_\RI$ with an application of rule $\RI 2 \C$. 
  Note that sequents in phase $\vdash_\C$ are not marked by list of tags, since after the application of right non-invertible rules there is no need to further permute formulae in context. Moreover, no new formulae can appear in context via applications of rule $\tl$, since the stoup is irreducible at this point.

  As already mentioned, rule $\tl$ has been modified. Its premise is now a sequent in phase $\vdash_\C$, which allows a further application of $\ex$ for the relocation of the formula $B$ to a different position in the context.

\subsection{Linear implication}\label{subsec:impl}
Finally, we consider a deductive system for a skew version of Lambek calculus with additive conjunction and disjunction.
This is obtained by extending the sequent calculus in (\ref{eq:seqcalc}) with a linear implication $\lolli$ and two introduction rules:
\begin{displaymath}
  \begin{array}{c}
    \infer[\lleft]{A \lolli B \mid \Gamma , \Delta \vdash C}{
      - \mid \Gamma \vdash A
      &
      B \mid \Delta \vdash C
    }
    \qquad
    \infer[\lright]{S \mid \Gamma \vdash A \lolli B}{S \mid \Gamma , A \vdash B}
  \end{array}
\end{displaymath}
The presence of $\lolli$ requires the extension of the congruence relation $\circeq$ with additional generating equations: an $\eta$-conversion and more permutative conversions.
\begin{equation*}\label{fig:circeq:impl}
  \small
  \begin{array}{rlll}
  \ax_{A \lolli B} &\circeq \lright \text{ } (\lleft \text{ } (\pass \text{ } \ax_{A}, \ax_{B} ))
  \\
  \tr \text{ } (\lleft \text{ } (f , g), h) & \circeq \lleft \text{ } (f, \tr \text{ } (g, h)) &&(f: {-} \mid \Gamma \vdash A', g : B' \mid \Delta \vdash A, h : {-} \mid \Lambda \vdash B)
  \\
  \pass \text{ } (\lright \text{ } f) &\circeq  \lright \text{ } (\pass \text{ } f) &&(f : A' \mid \Gamma , A \vdash B)
  \\
  \unitl \text{ } (\lright \text{ } f) &\circeq \lright \text{ } (\unitl \text{ } f) &&(f : {-} \mid \Gamma , A \vdash B)
  \\
  \tl \text{ } (\lright \text{ } f) &\circeq \lright \text{ } (\tl \text{ } f) &&(f : A' \mid B' , \Gamma , A \vdash B)
  \\
  \lleft \text{ } (f, \lright \text{ } g) &\circeq \lright \text{ } (\lleft \text{ } (f, g)) &&(f : {-} \mid \Gamma \vdash A', g : B' \mid \Delta , A \vdash B)
  \\
  \andli \ (\lright \ f) &\circeq \lright \ (\andli \ f) &&(f : A' \mid \Gamma , A \vdash B)
  \\
  \orl \ (\lright \ f , \lright \ g) &\circeq \lright \ (\orl \ (f ,g)) &&(f : A' \mid \Gamma , A \vdash B , g : B' \mid \Gamma , A \vdash B)
  \\
  \orri \ (\lleft \  (f , g)) &\circeq \lleft \ (f , \orri \ g) &&(f : - \mid \Gamma \vdash A , g : B \mid \Delta \vdash A')
  \end{array}
\end{equation*}

The sequent calculus enjoys categorical semantics in skew monoidal categories with binary products and coproducts, which moreover are endowed with a \emph{closed structure}, i.e. a functor $\lolli : \mathbb{C}^{\mathsf{op}} \times \mathbb{C} \rightarrow \mathbb{C}$ forming an adjunction ${-} \ot B \dashv B \lolli {-}$ for all objects $B$ \cite{street:skew-closed:2013}. There is no need to require left-distributivity, since this can now be proved using the adjunction and the universal property of coproducts.

Notice that, in non-commutative linear logic, there exist two distinct linear implications, also called left and right residuals \cite{lambek:mathematics:58}. Our calculus includes a single implication $\lolli$. We currently do not know whether the inclusion of the second implication to our logic is a meaningful addition nor whether it corresponds to some particular categorical notion. 

We now discuss the extension of the focused sequent calculus. This is more complicated than the extensions considered in Sections \ref{subsec:AddUnits} and \ref{subsec:Ex}. In order to understand the increased complexity, let us include the two new rules $\lright$ and $\lleft$ in the ``naive'' focused sequent calculus in (\ref{eq:naive:focus}). The right $\lolli$-rule is invertible, so it belongs to phase $\vdash_\RI$, while the left rule is not, so it goes in phase $\vdash_\F$.
\begin{displaymath}
  \begin{array}{c}
    \infer[\lleft]{A \lolli B \mid \Gamma , \Delta \vdash_\F P}{
      - \mid \Gamma \vdash_\RI A
      &
      B \mid \Delta \vdash_\LI P
    }
    \qquad
    \infer[\lright]{S \mid \Gamma \vdash_\RI A \lolli B}{S \mid \Gamma , A \vdash_\RI B}
  \end{array}
\end{displaymath}
As we know, this calculus is too permissive, and the inclusion of the above rules increases the non-deterministic choices in proof search even further. As a strategy for taming non-determinism, as before we decide to prioritize left non-invertible rules over right non-invertible ones. So we need to think of all possible situations when a right non-invertible rule must be applied before a left non-invertible one. The presence of $\lolli$ creates two new possibilities: $(i)$ $\lleft$ splits the context differently in different premises, or $(ii)$ left non-invertible rules manipulate formulae that have been moved to the context by applications of $\lright$. To understand these situations, let us look at two examples.

As an example of situation $(i)$, consider the sequent $\I \lolli \I \mid \I , Y \vdash_\F (\I \land \I) \ot Y$ and the following proof:
\begin{equation}\label{eq:ex:lleft:NonDeter}
  \small
  \begin{array}{c}
    \infer[\tr]{\I \lolli \I \mid \I , Y \vdash_\F (\I \land \I) \ot Y}{
      \infer[\andr]{\I \lolli \I \mid \I \vdash_\RI \I \land \I}{
  \infer[\sw]{\I \lolli \I \mid \I \vdash_\RI \I}{
    \infer[\lleft]{\I \lolli \I \mid \I \vdash_\F \I}{
      \infer[\sw]{- \mid \I \vdash_\RI \I}{
      \infer[\pass]{- \mid \I \vdash_\F \I}{
        \infer[\unitl]{\I \mid \quad \vdash_\LI \I}{
          \infer[\F 2 \LI]{- \mid \quad \vdash_\LI \I}{
          \infer[\unitr]{- \mid \quad \vdash_\F \I}{}}
        }
      }}
      &
      \infer[\unitl]{\I \mid \quad \vdash_\LI \I}{
        \infer[\F 2 \LI]{- \mid \quad \vdash_\LI \I}{
        \infer[\unitr]{- \mid \quad \vdash_\F \I}{}
      }}
    }
      }        
    &
   \infer[\sw]{\I \lolli \I \mid \I \vdash_\RI \I}{
        \infer[\lleft]{\I \lolli \I \mid \I \vdash_\F \I}{
      \infer[\sw]{- \mid \quad \vdash_\RI \I}{
      \infer[\unitr]{- \mid \quad \vdash_\F \I}{}}
      &
      \infer[\unitl]{\I \mid \I \vdash_\LI \I}{
        \infer[\F 2 \LI]{- \mid \I \vdash_\LI \I}{
        \infer[\pass]{- \mid \I \vdash_\F \I}{
          \infer[\unitl]{\I \mid \quad \vdash_\LI \I}{
            \infer[\F 2 \LI]{- \mid \quad \vdash_\LI \I}{
            \infer[\unitr]{- \mid \quad \vdash_\F \I}{}}
          }
        }
      }
    }}}}    
      &
      \infer[\sw]{- \mid Y \vdash_\RI Y}{
      \infer[\pass]{- \mid Y \vdash_\F Y}{
        \infer[\F 2 \LI]{Y \mid \quad \vdash_\LI Y}{
        \infer[\ax]{Y \mid \quad \vdash_\F Y}{}}
      }}
    }
  \end{array}
\end{equation}
Here $\tr$ must be applied first, before $\lleft$. In the proofs of the two premises of $\andr$, which prove the same sequent $\I \lolli \I \mid \I \vdash_\RI \I$, rule $\lleft$ splits the context in different ways: in the left branch the unit $\I$ in context is sent to the left premise, while in the right branch it goes to the right premise. If the application of the rule $\lleft$ would have split the context in the same way, then we could have applied $\lleft$ before $\tr$.

For $(ii)$, consider sequents ${-} \mid Y \vdash_{\F} (X \lolli X) \ot Y$ and $X \lolli Y \mid Z \vdash_{\F} (X \lolli Y) \ot Z$ with proofs:
\begin{equation}\label{eq:counterexample1}
  \small
    \proofbox{
       \infer[\tr]{{-} \mid Y \vdash_{\F} (X \lolli X) \ot Y}{
      \infer[\lright]{{-} \mid \quad \vdash_{\RI} X \lolli X}{
        \infer[\mathsf{sw}]{{-} \mid X \vdash_{\RI} X}{
          \infer[\pass]{{-} \mid X \vdash_{\F} X}{
            \infer[\F 2 \LI]{X \mid \quad \vdash_{\LI} X}{
              \infer[\ax]{X \mid \quad \vdash_{\F} X}{}
            }
          }
        }
      }
      &
      \infer[\mathsf{sw}]{{-} \mid Y \vdash_{\RI} Y}{
        \infer[\pass]{{-} \mid Y \vdash_{\F} Y}{
          \infer[\F 2 \LI]{Y \mid \quad \vdash_{\LI} Y}{
            \infer[\ax]{Y \mid \quad \vdash_{\F} Y}{}
          }
        }
      }
       }
       }
  \quad
  \proofbox{
     \infer[\tr]{X \lolli Y \mid Z \vdash_{\F} (X \lolli Y) \ot Z}{
      \infer[\lright]{X \lolli Y \mid \quad \vdash_{\RI} X \lolli Y}{
        \infer[\mathsf{sw}]{X \lolli Y \mid X \vdash_{\RI} Y}{
          \infer[\lleft]{X \lolli Y \mid X \vdash_{\F} Y}{
            \infer[\mathsf{sw}]{{-} \mid X \vdash_{\RI} X}{
              \infer[\pass]{{-} \mid X \vdash_{\F} X}{
                \infer[\F 2 \LI]{X \mid \quad \vdash_{\LI} X}{
                  \infer[\ax]{X \mid \quad \vdash_{\F} X}{}
                }
              }
            }
            &
            \infer[\F 2 \LI]{Y \mid \quad \vdash_{\LI} Y}{
              \infer[\ax]{Y \mid \quad \vdash_{\F} Y}{}
            }
          }
        }
      }
      &
      \infer[\mathsf{sw}]{{-} \mid Z \vdash_{\RI} Z}{
        \infer[\pass]{{-} \mid Z \vdash_{\F} Z}{
          \infer[\F 2 \LI]{Z \mid \quad \vdash_{\LI} Z}{
            \infer[\ax]{Z \mid \quad \vdash_{\F} Z}{}
          }
        }
      }
     }
     }
\end{equation}
In the first derivation, rule $\pass$ in the left branch of $\tr$ cannot be moved to the bottom of the proof tree, since formula $X$ is not yet in context, it becomes available only after the application of $\lright$. Analogously, in the second derivation, rule $\lleft$ in the left branch of $\tr$ cannot be moved at the bottom, since the formula $X$ that it sends to the left premise appears in context only after the application of $\lright$.

This motivates the addition of two new tags, corresponding to the two situations previously discussed: on top of $\tP,\tCone,\tCtwo$ and $\tE$, a tag could either be of the form $\Gamma$, for each context $\Gamma$, or of the form $\bullet$.  The validity condition for list of tags needs to be updated. A list of tags $l$ is now valid if it is non-empty and either $(i)$ $\tE \in l$, $(ii)$ both $\tCone \in l$ and $\tCtwo \in l$,  $(iii)$ there exist contexts $\Gamma , \Gamma'$ such that $ \Gamma \in l$, $ \Gamma' \in l$ and $ \Gamma \neq  \Gamma'$, or $(iv)$ $\bullet \in l$. Following \cite{UVW:protsn}, on top of tag annotations for sequents, we also require tag annotations for formulae in context. There is only one tag $\bullet$ for formulae. The tag on the formula $A^\bullet$ means that $A$ has been previously moved to the context by an application of $\lright$ in phase $\vdash^l_\RI$.

Here are the inference rules of the focused sequent calculus with linear implication:
\begin{equation}\label{eq:focus:impl}
  \!\!\!\!\!
  \begin{array}{lc}
    \text{(right invertible)} &
    \proofbox{
      \infer[\andr]{S \mid \Gamma \vdash^{l_1?, l_2?}_{\RI} A \land B}{
        S \mid \Gamma \vdash^{l_1?}_{\RI} A
        &
        S \mid \Gamma \vdash^{l_2?}_{\RI} B
      }
    \quad
    \infer[\lright]{S \mid \Gamma \vdash^{l?}_{\RI} A \lolli B}{S \mid \Gamma , A^{\bullet ?} \vdash^{l?}_{\RI} B}
    \quad
    \infer[\LI 2 \RI]{S \mid \Gamma \vdash^{t?}_{\RI} P}{S \mid \Gamma \vdash^{t?}_{\LI} P}
    }
    \\[10pt]
    \text{(left invertible)} &
    \proofbox{
      \infer[\unitl]{\I \mid \Gamma \vdash_{\LI} P}{{-} \mid \Gamma \vdash_{\LI} P}
    \qquad
    \infer[\tl]{A \ot B \mid \Gamma \vdash_{\LI} P}{A \mid B , \Gamma \vdash_{\LI} P}
    \\
    \infer[\orl]{A \lor B \mid \Gamma \vdash_\LI P}{
      A \mid \Gamma \vdash_\LI P
      &
      B \mid \Gamma \vdash_\LI P
    }
    \qquad
    \infer[\F 2 \LI]{T \mid \Gamma \vdash^{t?}_{\LI} P}{T \mid \Gamma \vdash^{t?}_{\F} P}
    }
    \\[10pt]
    \text{(focusing)} & 
    \proofbox{
    \infer[\pass]{{-} \mid A^{\bullet ?} , \Gamma \vdash^{t?}_{\F} P }{
        A \mid \Gamma^{\circ} \vdash_{\LI} P
        &
        \text{if } A^{\bullet?} = A \text{ then } (t \text{ does not exist or } t = \tP) \text{ else } t = \bullet
    }
    }
    \\[10pt]
    \multicolumn{2}{c}{
    \infer[\ax]{X \mid \quad \vdash^{\tE?}_{\F} X}{}
    \qquad
    \infer[\unitr]{{-} \mid \quad \vdash^{\tE?}_{\F} \I}{}
    \qquad
    \infer[\andlone]{A \land B \mid \Gamma \vdash^{\tCone?}_{\F} P}{A \mid \Gamma^{\circ} \vdash_{\LI} P}
    \qquad
    \infer[\andltwo]{A \land B \mid \Gamma \vdash^{\tCtwo?}_{\F} P}{B \mid \Gamma^{\circ} \vdash_{\LI} P}
    }
    \\[10pt]
    \multicolumn{2}{c}{
    \infer[\tr]{T \mid \Gamma , \Delta \vdash^{\tE?}_{\F} A \ot B}{
      T \mid \Gamma^{\circ} \vdash^{l}_{\RI} A
      &
      {-} \mid \Delta^{\circ} \vdash_{\RI} B
      &
      l \ \text{valid}
    }
    \;
    \infer[\orrone]{T \mid \Gamma \vdash^{\tE?}_{\F} A \lor B}{
      T \mid \Gamma^{\circ} \vdash^{l}_{\RI} A
      &
      l \ \text{valid}
    }
    \;
    \infer[\orrtwo]{T \mid \Gamma \vdash^{\tE?}_{\F} A \lor B}{
      T \mid \Gamma^{\circ} \vdash^{l}_{\RI} B
      &
      l \ \text{valid}
    }
    }
    \\[6pt]
    \multicolumn{2}{c}{
      \infer[\lleft]{A \lolli B \mid \Gamma , \Delta^{\bullet} , \Lambda \vdash^{t?}_{\F} P}{
        - \mid \Gamma , \Delta^{\circ} \vdash_{\RI} A
        &
        B \mid \Lambda^{\circ} \vdash_{\LI} P
        &
        \text{if } \Delta^{\bullet} \text{ is empty then } (t \text{ does not exist or } t =  \Gamma)  \text{ else }  t = \bullet 
      }
    }
  \end{array}
\end{equation}
Again $P$ indicates a non-negative formula, which now means that its principal connective is neither $\land$ nor $\lolli$. The notation $\Gamma^{\bullet}$ means that all the formulae in $\Gamma$ are tagged, while $\Gamma^{\circ}$ indicates that all the tags on formulae in $\Gamma$ have been erased. We write $A^{\bullet?}$ to denote $A$ if the formula appears in an untagged sequent and $A^\bullet$ if it appears in a sequent marked with a list of tags $l$ or a single tag $t$.

Tags of the form $t = \Gamma$ are used to record different splitting of context in applications of $\lleft$, while tag $t = \bullet$ marks when rule $\lleft$ sends tagged formulae to the left premise and when rule $\pass$ moves a tagged formula to the stoup.

Rule $\lright$  moves a formula $A$ from the succedent to the right end of the context. If its conclusion is marked by a list of tags $l$, then $A$ is also tagged with $\bullet$.

The side condition in rule $\lleft$ should be read as follows. The tagged context $\Delta^{\bullet}$ starts with the leftmost tagged formula in the sequent. If $\Delta^{\bullet}$ is empty, then the sequent is either untagged (so there is no $t$) or the tag $t$ is equal to $ \Gamma$. If $\Delta^{\bullet}$ is non-empty, then $t = \bullet$. In particular,
$\Delta^{\bullet}$ contains at least one tagged formula, which must have appeared in context from an application of $\lright$. If $\Delta^{\bullet}$ is empty and $t = \Gamma$, no new (meaning: tagged with $\bullet$) formula is moved to the left premise. If $t = \Gamma$ then we are performing proof search inside the premise of a right non-invertible rule and $t$ belongs to some valid list of tags $l$. List $l$ could be valid because of a different branch in the proof tree where $\lleft$ is also applied but the context has been split differently (so its tag would be $ \Gamma'$ for some $\Gamma \not= \Gamma'$).

Rule $\pass$ has a similar side condition to $\lleft$. If $A$ does not have a tag, then the sequent is also untagged or the tag $t$ is equal to $\tP$. If $A$ has tag $\bullet$, then $t$ must also be $\bullet$. In other words, if $t = \bullet$ then the formula that $\pass$ moves to the stoup must also be tagged with $\bullet$, i.e. must have been added to the context by an application of $\lright$.

We can reconstruct the derivation in (\ref{eq:ex:lleft:NonDeter}) within the focused sequent calculus with tags in (\ref{eq:focus:impl}).
\begin{displaymath}\small
  \infer[\sw]{\I \lolli \I \mid \I , Y \vdash_{\RI} (\I \land \I) \ot Y}{
    \infer[\tr]{\I \lolli \I \mid \I , Y \vdash_{\F} (\I \land \I) \ot Y}{
      \infer[\andr]{\I \lolli \I \mid \I \vdash^{{[\ ]} , {[\I]}}_{\RI} \I \land \I}{
        \infer[\sw]{\I \lolli \I \mid \I \vdash^{{[\ ]}}_{\RI} \I}{
          \infer[\lleft]{\I \lolli \I \mid \I \vdash^{{[\ ]}}_{\F} \I}{
            \infer[\sw]{- \mid \quad \vdash_{\RI} \I}{
              \infer[\unitr]{- \mid \quad \vdash_{\F} \I}{}
            }
            &
            \infer[\unitl]{\I \mid \I \vdash_{\LI} \I}{
              \infer[\F 2 \LI]{- \mid \I \vdash_{\LI} \I}{
                \infer[\pass]{- \mid \I \vdash_{\F} \I}{
                  \infer[\unitl]{\I \mid \quad \vdash_{\LI} \I}{
                    \infer[\F 2 \LI]{- \mid \quad \vdash_{\LI} \I}{
                      \infer[\unitr]{- \mid \quad \vdash_{\F} \I}{}
                    }
                  }
                }
              }
            }
          }
        }
        &
        \infer[\sw]{\I \lolli \I \mid \I \vdash^{{[\I]}}_{\RI} \I}{
          \infer[\lleft]{\I \lolli \I \mid \I \vdash^{{[\I]}}_{\F} \I}{
            \infer[\sw]{- \mid \I \vdash_{\RI} \I}{
              \infer[\pass]{- \mid \I \vdash_{\F} \I}{
                \infer[\unitl]{\I \mid \quad \vdash_{\LI} \I}{
                  \infer[\F 2 \LI]{- \mid \quad \vdash_{\LI} \I}{
                    \infer[\unitr]{- \mid \quad \vdash_{\F} \I}{}
                  }
                }
              }
            }
            &
            \infer[\unitl]{\I \mid \quad \vdash_{\LI} \I}{
              \infer[\F 2 \LI]{- \mid \quad \vdash_{\LI} \I}{
                \infer[\unitr]{- \mid \quad \vdash_{\F} \I}{}
              }
            }
          }
        }
      }
      &
      \infer[\sw]{- \mid Y \vdash_{\RI} Y}{
        \infer[\pass]{- \mid Y \vdash_{\F} Y}{
          \infer[\F 2 \LI]{Y \mid \quad \vdash_{\LI} Y}{
            \infer[\ax]{Y \mid \quad \vdash_{\F} Y}{}
          }
        }
      }
    }
  }
\end{displaymath}
The proofs with tags of the derivations in (\ref{eq:counterexample1}) are analogous to the ones described in \cite{UVW:protsn}.

Proving completeness of the extended focused sequent calculus is more involved than in the absence of implication. Concretely, the complication resides in stating and proving the analog of Proposition \ref{prop:GenRightRules}. First, define an operation $\mathsf{impconj}(A)$ which produces a list of pairs of lists of formulae and formulae as follows: 
\begin{displaymath}
  \begin{array}{rll}
    \impconj{A} &= \impconj{A'} , \impconj{B'} &\text{when } A = A' \land B'
    \\
    \impconj{A} &= ((A' , \Gamma'_1) , B'_1) , \dots , ((A' , \Gamma'_n) , B'_n) &\text{when } A = A' \lolli B' \text{ and}
    \\
    & &\impconj{B'} = ([(\Gamma'_1 , B'_1) , \dots , (\Gamma'_n , B'_n)])
    \\
    \impconj{A} &= ([\ ] , A) &\text{otherwise}
  \end{array}
\end{displaymath}
For example, $\impconj{A \lolli (B \lolli (X \land (C \lor D) \land (Y \lolli Z)))} = [([A , B] , X) , ([A , B] , C \lor D) , ([A , B , Y] , Z)]$.

The statement of Proposition \ref{prop:GenRightRules} for the focused sequent calculus in (\ref{eq:focus:impl}) then becomes:
\begin{proposition}\label{prop:GenRightRules:impl}
  The following rules
  \begin{displaymath}
    \begin{array}{c}
      \infer[\orrone^{\LI}]{S \mid \Gamma \vdash_{\LI} A \lor B}{
        \deduce{[S \mid \Gamma , \Gamma'_i \vdash_{\LI} P_i]_{i \in [1 , \dots , n]}}{fs}
      }
      \quad
      \infer[\orrtwo^{\LI}]{S \mid \Gamma \vdash_{\LI} A \lor B}{
        \deduce{[S \mid \Gamma , \Gamma''_i \vdash_{\LI} Q_i]_{i \in [1 , \dots , m]}}{fs}
      }
      \\
      \infer[\tr^{\LI}]{S \mid \Gamma , \Delta \vdash_{\LI} A \ot B'}{
        \deduce{[S \mid \Gamma , \Gamma'_i \vdash_{\LI} P_i]_{i \in [1 , \dots , n]}}{fs}
        &
        - \mid \Delta \vdash_{\RI} B'
      }
    \end{array}
  \end{displaymath}
  are admissible, where
  $\impconj{A} = [(\Gamma'_1 , P_1) , \dots , (\Gamma'_n , P_n)]$ and $\impconj{B} = [(\Gamma''_1 , Q_1) , \dots , (\Gamma''_m , Q_m)]$.
\end{proposition}

\section{Conclusion}
The paper presents a sequent calculus for a semi-associative and semi-unital logic,  extending the system introduced in \cite{uustalu:sequent:2021} with additive conjunction and disjunction. Categorical models of this calculus are skew monoidal categories with binary products and coproducts, and the tensor product preserves coproducts on the left: $(A + B) \ot C \cong (A \ot C) + (B \ot C)$.
Derivations in the sequent calculus are equated by a congruence relation $\circeq$ and canonical representatives of each $\circeq$-equivalence class can be computed in a separate sequent calculus of normal forms, that we dubbed ``focused'' due to its phase separation similar to the one in Andreoli's technique \cite{andreoli:logic:1992}.  It should be remarked that, differently from Andreoli's focusing, and also the maximally multi-focused sequent calculus for skew monoidal closed categories by one of the authors \cite{veltri:multifocus:23}, we do not insist on keeping the focus during the synchronous phase of proof search, and we always privilege the application of left non-invertible rules over right non-invertible ones. In order to achieve completeness wrt. the sequent calculus, the focused system employs a system of tag annotations providing explicit justifications for cases where right non-invertible rules must be applied before the left non-invertible ones.

The focused sequent calculus is a concrete presentation of the free distributive skew monoidal category on the set of atomic formulae. Therefore the normalization/focusing algorithm determines a procedure for solving the coherence problem of distributive skew monoidal categories.

In the final part of the paper, we have looked at extensions of the logic with additive units, a skew exchange rule in the style of Bourke and Lack \cite{bourke:lack:braided:2020}, and linear implication. This section still needs to be formalized in Agda, which will be our forthcoming step.

This paper takes one step further in a large project aiming at modularly analyzing proof systems with categorical models given by categories with skew structure \cite{zeilberger:semiassociative:19, uustalu:sequent:2021,uustalu:proof:nodate,uustalu:deductive:nodate,veltri:coherence:2021,UVW:protsn}. We are interested in looking for applications of these systems to combinatorics and linguistics, following the initial investigation by Zeilberger \cite{zeilberger:semiassociative:19} and Moortgat \cite{moortgat:tamari:20}.

We are interested in using the techniques introduced in this paper to design a calculus of normal forms for the classical linear logic MALL.
In the latter setting, the situation is more complicated than the skew one, since there could be more than two formulae that can be under focus in the synchronous phase.
We expect our calculus to give an alternative presentation of the maximally multi-focused proofs of Chaudhuri et al. \cite{chaudhuri:canonical:2008}.

\paragraph{Acknowledgements}
This work was supported by the Estonian Research Council grant PSG749 and the ESF funded Estonian IT Academy research measure (project 2014-2020.4.05.19-0001). 

  \bibliographystyle{eptcs}
  \bibliography{LSFA}

\begin{thebibliography}{10}
\providecommand{\bibitemdeclare}[2]{}
\providecommand{\surnamestart}{}
\providecommand{\surnameend}{}
\providecommand{\urlprefix}{Available at }
\providecommand{\url}[1]{\texttt{#1}}
\providecommand{\href}[2]{\texttt{#2}}
\providecommand{\urlalt}[2]{\href{#1}{#2}}
\providecommand{\doi}[1]{doi:\urlalt{https://doi.org/#1}{#1}}
\providecommand{\eprint}[1]{arXiv:\urlalt{https://arxiv.org/abs/#1}{#1}}
\providecommand{\bibinfo}[2]{#2}

\bibitemdeclare{article}{abrusci:noncommutative:1990}
\bibitem{abrusci:noncommutative:1990}
\bibinfo{author}{Vito~Michele \surnamestart Abrusci\surnameend} (\bibinfo{year}{1990}): \emph{\bibinfo{title}{Non-commutative Intuitionistic Linear Logic}}.
\newblock {\slshape \bibinfo{journal}{Mathematical Logic Quarterly}} \bibinfo{volume}{36}(\bibinfo{number}{4}), pp. \bibinfo{pages}{297--318}, \doi{10.1002/malq.19900360405}.

\bibitemdeclare{article}{altenkirch:monads:2014}
\bibitem{altenkirch:monads:2014}
\bibinfo{author}{Thorsten \surnamestart Altenkirch\surnameend}, \bibinfo{author}{James \surnamestart Chapman\surnameend} \& \bibinfo{author}{Tarmo \surnamestart Uustalu\surnameend} (\bibinfo{year}{2015}): \emph{\bibinfo{title}{Monads Need Not Be Endofunctors}}.
\newblock {\slshape \bibinfo{journal}{Logical Methods in Computer Science}} \bibinfo{volume}{11}(\bibinfo{number}{1}):\bibinfo{eid}{3}, \doi{10.2168/lmcs-11(1:3)2015}.

\bibitemdeclare{article}{andreoli:logic:1992}
\bibitem{andreoli:logic:1992}
\bibinfo{author}{Jean-Marc \surnamestart Andreoli\surnameend} (\bibinfo{year}{1992}): \emph{\bibinfo{title}{Logic Programming with Focusing Proofs in Linear Logic}}.
\newblock {\slshape \bibinfo{journal}{Journal of Logic and Computation}} \bibinfo{volume}{2}(\bibinfo{number}{3}), pp. \bibinfo{pages}{297--347}, \doi{10.1093/logcom/2.3.297}.

\bibitemdeclare{article}{bourke:skew:2017}
\bibitem{bourke:skew:2017}
\bibinfo{author}{John \surnamestart Bourke\surnameend} (\bibinfo{year}{2017}): \emph{\bibinfo{title}{Skew structures in 2-category theory and homotopy theory}}.
\newblock {\slshape \bibinfo{journal}{Journal of Homotopy and Related Structures}} \bibinfo{volume}{12}(\bibinfo{number}{1}), pp. \bibinfo{pages}{31--81}, \doi{10.1007/s40062-015-0121-z}.

\bibitemdeclare{article}{bourke:skew:2018}
\bibitem{bourke:skew:2018}
\bibinfo{author}{John \surnamestart Bourke\surnameend} \& \bibinfo{author}{Stephen \surnamestart Lack\surnameend} (\bibinfo{year}{2018}): \emph{\bibinfo{title}{Skew Monoidal Categories and Skew Multicategories}}.
\newblock {\slshape \bibinfo{journal}{Journal of Algebra}} \bibinfo{volume}{506}, pp. \bibinfo{pages}{237--266}, \doi{10.1016/j.jalgebra.2018.02.039}.

\bibitemdeclare{article}{bourke:lack:braided:2020}
\bibitem{bourke:lack:braided:2020}
\bibinfo{author}{John \surnamestart Bourke\surnameend} \& \bibinfo{author}{Stephen \surnamestart Lack\surnameend} (\bibinfo{year}{2020}): \emph{\bibinfo{title}{Braided Skew Monoidal Categories}}.
\newblock {\slshape \bibinfo{journal}{Theory and Applications of Categories}} \bibinfo{volume}{35}(\bibinfo{number}{2}), pp. \bibinfo{pages}{19--63}.
\newblock \urlprefix\url{http://www.tac.mta.ca/tac/volumes/35/2/35-02abs.html}.

\bibitemdeclare{article}{buckley:catalan:2015}
\bibitem{buckley:catalan:2015}
\bibinfo{author}{Michael \surnamestart Buckley\surnameend}, \bibinfo{author}{Richard \surnamestart Garner\surnameend}, \bibinfo{author}{Stephen \surnamestart Lack\surnameend} \& \bibinfo{author}{Ross \surnamestart Street\surnameend} (\bibinfo{year}{2015}): \emph{\bibinfo{title}{The {C}atalan Simplicial Set}}.
\newblock {\slshape \bibinfo{journal}{Mathematical Proceedings of Cambridge Philosophical Society}} \bibinfo{volume}{158}(\bibinfo{number}{2}), pp. \bibinfo{pages}{211--222}, \doi{10.1017/s0305004114000498}.

\bibitemdeclare{inproceedings}{chaudhuri:canonical:2008}
\bibitem{chaudhuri:canonical:2008}
\bibinfo{author}{Kaustuv \surnamestart Chaudhuri\surnameend}, \bibinfo{author}{Dale \surnamestart Miller\surnameend} \& \bibinfo{author}{Alexis \surnamestart Saurin\surnameend} (\bibinfo{year}{2008}): \emph{\bibinfo{title}{Canonical Sequent Proofs via Multi-Focusing}}.
\newblock In \bibinfo{editor}{Giorgio \surnamestart Ausiello\surnameend}, \bibinfo{editor}{Juhani \surnamestart Karhum{\"{a}}ki\surnameend}, \bibinfo{editor}{Giancarlo \surnamestart Mauri\surnameend} \& \bibinfo{editor}{Luke \surnamestart Ong\surnameend}, editors: {\slshape \bibinfo{booktitle}{Proceedings of 5th IFIP International Conference on Theoretical Computer Science, {TCS} 2008}}, {\slshape \bibinfo{series}{International\ Federation of Information Processing Series}} \bibinfo{volume}{273}, \bibinfo{publisher}{Springer}, pp. \bibinfo{pages}{383--396}, \doi{10.1007/978-0-387-09680-3\_26}.

\bibitemdeclare{article}{girard:linear:87}
\bibitem{girard:linear:87}
\bibinfo{author}{Jean{-}Yves \surnamestart Girard\surnameend} (\bibinfo{year}{1987}): \emph{\bibinfo{title}{Linear Logic}}.
\newblock {\slshape \bibinfo{journal}{Theoretical Computer Science}} \bibinfo{volume}{50}, pp. \bibinfo{pages}{1--102}, \doi{10.1016/0304-3975(87)90045-4}.

\bibitemdeclare{article}{girard:constructive:91}
\bibitem{girard:constructive:91}
\bibinfo{author}{Jean-Yves \surnamestart Girard\surnameend} (\bibinfo{year}{1991}): \emph{\bibinfo{title}{A New Constructive Logic: Classical Logic}}.
\newblock {\slshape \bibinfo{journal}{Mathematical Structures in Computer Science}} \bibinfo{volume}{1}(\bibinfo{number}{3}), pp. \bibinfo{pages}{255--296}, \doi{10.1017/s0960129500001328}.

\bibitemdeclare{article}{lack:skew:2012}
\bibitem{lack:skew:2012}
\bibinfo{author}{Stephen \surnamestart Lack\surnameend} \& \bibinfo{author}{Ross \surnamestart Street\surnameend} (\bibinfo{year}{2012}): \emph{\bibinfo{title}{Skew Monoidales, Skew Warpings and Quantum Categories}}.
\newblock {\slshape \bibinfo{journal}{Theory and Applications of Categories}} \bibinfo{volume}{26}, pp. \bibinfo{pages}{385--402}.
\newblock \urlprefix\url{http://www.tac.mta.ca/tac/volumes/26/15/26-15abs.html}.

\bibitemdeclare{article}{lack:triangulations:2014}
\bibitem{lack:triangulations:2014}
\bibinfo{author}{Stephen \surnamestart Lack\surnameend} \& \bibinfo{author}{Ross \surnamestart Street\surnameend} (\bibinfo{year}{2014}): \emph{\bibinfo{title}{Triangulations, Orientals, and Skew Monoidal Categories}}.
\newblock {\slshape \bibinfo{journal}{Advances in Mathematics}} \bibinfo{volume}{258}, pp. \bibinfo{pages}{351--396}, \doi{10.1016/j.aim.2014.03.003}.

\bibitemdeclare{article}{lambek:mathematics:58}
\bibitem{lambek:mathematics:58}
\bibinfo{author}{Joachim \surnamestart Lambek\surnameend} (\bibinfo{year}{1958}): \emph{\bibinfo{title}{The Mathematics of Sentence Structure}}.
\newblock {\slshape \bibinfo{journal}{American Mathematical Monthly}} \bibinfo{volume}{65}(\bibinfo{number}{3}), pp. \bibinfo{pages}{154--170}, \doi{10.2307/2310058}.

\bibitemdeclare{article}{maclane1963natural}
\bibitem{maclane1963natural}
\bibinfo{author}{Saunders \surnamestart {Mac Lane}\surnameend} (\bibinfo{year}{1963}): \emph{\bibinfo{title}{Natural Associativity and Commutativity}}.
\newblock {\slshape \bibinfo{journal}{Rice University Studies}} \bibinfo{volume}{49}(\bibinfo{number}{4}), pp. \bibinfo{pages}{28--46}.
\newblock \urlprefix\url{http://hdl.handle.net/1911/62865}.

\bibitemdeclare{misc}{moortgat:tamari:20}
\bibitem{moortgat:tamari:20}
\bibinfo{author}{Michael \surnamestart Moortgat\surnameend} (\bibinfo{year}{2020}): \emph{\bibinfo{title}{The {T}amari order for $D^3$ and derivability in semi-associative {L}ambek-{G}rishin {C}alculus}}.
\newblock \bibinfo{howpublished}{Talk at 16th Workshop on Computational Logic and Applications, CLA 2020}.
\newblock \bibinfo{note}{Slides available at: \url{http://cla.tcs.uj.edu.pl/history/2020/pdfs/CLA_slides_Moortgat.pdf}}.

\bibitemdeclare{book}{moot:logic:12}
\bibitem{moot:logic:12}
\bibinfo{author}{Richard \surnamestart Moot\surnameend} \& \bibinfo{author}{Christian \surnamestart Retor{\'{e}}\surnameend} (\bibinfo{year}{2012}): \emph{\bibinfo{title}{The Logic of Categorial Grammars - {A} Deductive Account of Natural Language Syntax and Semantics}}.
\newblock {\slshape \bibinfo{series}{Lecture Notes in Computer Science}} \bibinfo{volume}{6850}, \bibinfo{publisher}{Springer}, \doi{10.1007/978-3-642-31555-8}.

\bibitemdeclare{inproceedings}{scherer:simple:2015}
\bibitem{scherer:simple:2015}
\bibinfo{author}{Gabriel \surnamestart Scherer\surnameend} \& \bibinfo{author}{Ddier \surnamestart R{\'{e}}my\surnameend} (\bibinfo{year}{2015}): \emph{\bibinfo{title}{Which Simple Types Have a Unique Inhabitant?}}
\newblock In: {\slshape \bibinfo{booktitle}{Proceedings of 20th {ACM} {SIGPLAN} International Conference on Functional Programming, ICFP 2015}}, \bibinfo{publisher}{{ACM}}, pp. \bibinfo{pages}{243--255}, \doi{10.1145/2784731.2784757}.

\bibitemdeclare{article}{street:skew-closed:2013}
\bibitem{street:skew-closed:2013}
\bibinfo{author}{Ross \surnamestart Street\surnameend} (\bibinfo{year}{2013}): \emph{\bibinfo{title}{Skew-Closed Categories}}.
\newblock {\slshape \bibinfo{journal}{Journal of Pure and Applied Algebra}} \bibinfo{volume}{217}(\bibinfo{number}{6}), pp. \bibinfo{pages}{973--988}, \doi{10.1016/j.jpaa.2012.09.020}.

\bibitemdeclare{article}{szlachanyi:skew-monoidal:2012}
\bibitem{szlachanyi:skew-monoidal:2012}
\bibinfo{author}{Korn{\'e}l \surnamestart Szlach\'anyi\surnameend} (\bibinfo{year}{2012}): \emph{\bibinfo{title}{Skew-Monoidal Categories and Bialgebroids}}.
\newblock {\slshape \bibinfo{journal}{Advances in Mathematics}} \bibinfo{volume}{231}(\bibinfo{number}{3--4}), pp. \bibinfo{pages}{1694--1730}, \doi{10.1016/j.aim.2012.06.027}.

\bibitemdeclare{inproceedings}{UVW:protsn}
\bibitem{UVW:protsn}
\bibinfo{author}{Tarmo \surnamestart Uustalu\surnameend}, \bibinfo{author}{Niccol{\`o} \surnamestart Veltri\surnameend} \& \bibinfo{author}{Cheng-Syuan \surnamestart Wan\surnameend} (\bibinfo{year}{2022}): \emph{\bibinfo{title}{Proof Theory of Skew Non-Commutative MILL}}.
\newblock In \bibinfo{editor}{Andrzej \surnamestart Indrzejczak\surnameend} \& \bibinfo{editor}{Michal \surnamestart Zawidzki\surnameend}, editors: {\slshape \bibinfo{booktitle}{Proceedings of 10th International Conference on Non-classical Logics: Theory and Applications, {NCL} 2022}}, {\slshape \bibinfo{series}{Electronic Proceedings in Theoretical Computer Science}} \bibinfo{volume}{358}, \bibinfo{publisher}{Open Publishing Association}, pp. \bibinfo{pages}{118--135}, \doi{10.4204/eptcs.358.9}.

\bibitemdeclare{inproceedings}{uustalu:deductive:nodate}
\bibitem{uustalu:deductive:nodate}
\bibinfo{author}{Tarmo \surnamestart Uustalu\surnameend}, \bibinfo{author}{Niccol{\'o} \surnamestart Veltri\surnameend} \& \bibinfo{author}{Noam \surnamestart Zeilberger\surnameend} (\bibinfo{year}{2021}): \emph{\bibinfo{title}{Deductive Systems and Coherence for Skew Prounital Closed Categories}}.
\newblock In \bibinfo{editor}{Claudio \surnamestart {Sacerdoti Coen}\surnameend} \& \bibinfo{editor}{Alwen \surnamestart Tiu\surnameend}, editors: {\slshape \bibinfo{booktitle}{Proceedings of 15th Workshop on Logical Frameworks and Meta-Languages: Theory and Practice, {LFMTP} 2020}}, {\slshape \bibinfo{series}{Electronic Proceedings in Theoretical Computer Science}} \bibinfo{volume}{332}, \bibinfo{publisher}{Open Publishing Association}, pp. \bibinfo{pages}{35--53}, \doi{10.4204/eptcs.332.3}.

\bibitemdeclare{inproceedings}{uustalu:proof:nodate}
\bibitem{uustalu:proof:nodate}
\bibinfo{author}{Tarmo \surnamestart Uustalu\surnameend}, \bibinfo{author}{Niccol{\`o} \surnamestart Veltri\surnameend} \& \bibinfo{author}{Noam \surnamestart Zeilberger\surnameend} (\bibinfo{year}{2021}): \emph{\bibinfo{title}{Proof Theory of Partially Normal Skew Monoidal Categories}}.
\newblock In \bibinfo{editor}{David~I. \surnamestart Spivak\surnameend} \& \bibinfo{editor}{Jamie \surnamestart Vicary\surnameend}, editors: {\slshape \bibinfo{booktitle}{Proceedings of 3rd Annual International Applied Category Theory Conference 2020, {ACT} 2020}}, {\slshape \bibinfo{series}{Electronic Proceedings in Theoretical Computer Science}} \bibinfo{volume}{333}, \bibinfo{publisher}{Open Publishing Association}, pp. \bibinfo{pages}{230--246}, \doi{10.4204/eptcs.333.16}.

\bibitemdeclare{incollection}{uustalu:sequent:2021}
\bibitem{uustalu:sequent:2021}
\bibinfo{author}{Tarmo \surnamestart Uustalu\surnameend}, \bibinfo{author}{Niccol{\`o} \surnamestart Veltri\surnameend} \& \bibinfo{author}{Noam \surnamestart Zeilberger\surnameend} (\bibinfo{year}{2021}): \emph{\bibinfo{title}{The Sequent Calculus of Skew Monoidal Categories}}.
\newblock In \bibinfo{editor}{Claudio \surnamestart Casadio\surnameend} \& \bibinfo{editor}{Philip~J. \surnamestart Scott\surnameend}, editors: {\slshape \bibinfo{booktitle}{Joachim Lambek: The Interplay of Mathematics, Logic, and Linguistics}}, {\slshape \bibinfo{series}{Outstanding Contributions to Logic}}~\bibinfo{volume}{20}, \bibinfo{publisher}{Springer}, pp. \bibinfo{pages}{377--406}, \doi{10.1007/978-3-030-66545-6_11}.

\bibitemdeclare{inproceedings}{veltri:coherence:2021}
\bibitem{veltri:coherence:2021}
\bibinfo{author}{Niccol{\`o} \surnamestart Veltri\surnameend} (\bibinfo{year}{2021}): \emph{\bibinfo{title}{Coherence via Focusing for Symmetric Skew Monoidal Categories}}.
\newblock In \bibinfo{editor}{Alexandra \surnamestart Silva\surnameend}, \bibinfo{editor}{Renata \surnamestart Wassermann\surnameend} \& \bibinfo{editor}{Ruy \surnamestart {de Queiroz}\surnameend}, editors: {\slshape \bibinfo{booktitle}{Proceedings\ of 27th International\ Workshop on Logic, Language, Information, and Computation, WoLLIC 2021}}, {\slshape \bibinfo{series}{Lecture Notes in Computer Science}} \bibinfo{volume}{13028}, \bibinfo{publisher}{Springer}, pp. \bibinfo{pages}{184--200}, \doi{10.1007/978-3-030-88853-4_12}.

\bibitemdeclare{inproceedings}{veltri:multifocus:23}
\bibitem{veltri:multifocus:23}
\bibinfo{author}{Niccol{\`{o}} \surnamestart Veltri\surnameend} (\bibinfo{year}{2023}): \emph{\bibinfo{title}{Maximally Multi-focused Proofs for Skew Non-Commutative {MILL}}}.
\newblock In \bibinfo{editor}{Helle~Hvid \surnamestart Hansen\surnameend}, \bibinfo{editor}{Andre \surnamestart Scedrov\surnameend} \& \bibinfo{editor}{Ruy J. G.~B. \surnamestart de~Queiroz\surnameend}, editors: {\slshape \bibinfo{booktitle}{Proceedings of 29th International\ Workshop on Logic, Language, Information, and Computation, WoLLIC 2023}}, {\slshape \bibinfo{series}{Lecture Notes in Computer Science}} \bibinfo{volume}{13923}, \bibinfo{publisher}{Springer}, pp. \bibinfo{pages}{377--393}, \doi{10.1007/978-3-031-39784-4\_24}.

\bibitemdeclare{article}{zeilberger:semiassociative:19}
\bibitem{zeilberger:semiassociative:19}
\bibinfo{author}{Noam \surnamestart Zeilberger\surnameend} (\bibinfo{year}{2019}): \emph{\bibinfo{title}{A Sequent Calculus for a Semi-Associative Law}}.
\newblock {\slshape \bibinfo{journal}{Logical Methods in Computer Science}} \bibinfo{volume}{15}(\bibinfo{number}{1}):\bibinfo{eid}{9}, \doi{10.23638/lmcs-15(1:9)2019}.

\end{thebibliography}
\end{document}